\documentclass{imsart}

\usepackage[round]{natbib}
\usepackage{amsmath,amssymb,amsthm,bm,enumerate,mathrsfs,mathtools}
\usepackage{latexsym,color,verbatim,multirow}
\usepackage{graphicx}
\usepackage{caption}
\usepackage{subcaption}
\usepackage{mathtools}

\usepackage{graphics,amssymb,color}
\usepackage{graphicx}%,subfigure}
\usepackage{algorithm}
\usepackage{algorithm,algorithmic}
\usepackage{xfrac}
\usepackage{subcaption}
\usepackage{float}

\usepackage{hyperref}%%put last
\hypersetup{colorlinks,
citecolor=blue,
linkcolor=blue,
urlcolor=black}

\newcommand{\argmin}{\mathop{\mathrm{argmin}}}

\newcommand{\real}{\mathbb{R}}

\newcommand{\Ee}{{\mathbb E}}
\newcommand{\E}{\Ee}

\newtheorem{theorem}{Theorem}
\newtheorem{lemma}{Lemma}

\newtheorem{remark}{Remark}
\newtheorem{example}{Example}[section]

\newcommand{\grad}{\nabla}

\allowdisplaybreaks
\makeatletter
\newcommand*{\rom}[1]{\expandafter\@slowromancap\romannumeral #1@}
\makeatother

\newcommand*{\Scale}[2][4]{\scalebox{#1}{$#2$}}

\makeatletter
\renewcommand*\env@matrix[1][\arraystretch]{%
  \edef\arraystretch{#1}%
  \hskip -\arraycolsep
  \let\@ifnextchar\new@ifnextchar
  \array{*\c@MaxMatrixCols c}}
\makeatother

\usepackage{makecell}

\begin{document}

%\title{An optimization-based method for scalable Bayesian selective inference}
\title{Scalable methods for Bayesian selective inference}
\author{Snigdha Panigrahi, Jonathan Taylor}
\runauthor{Panigrahi and Taylor}
%\runtitle{Approximating a selective posterior}

\begin{abstract}
Modeled along the truncated approach in \cite{selective_bayesian}, selection-adjusted inference in a Bayesian regime is based on a \textit{selective posterior}. Such a posterior is determined together by a generative model imposed on data and the selection event that enforces a truncation on the assumed law. The effective difference between the \textit{selective posterior} and the usual Bayesian framework is reflected in the use of a truncated likelihood. The normalizer of the truncated law in the adjusted framework is the probability of the selection event; this typically lacks a closed form expression leading to the computational bottleneck in sampling from such a posterior. The current work provides an  optimization problem that approximates the otherwise intractable \textit{selective posterior} and leads to scalable methods that give valid post-selective Bayesian inference.
%The current work lays out a primal-dual approach of solving an approximating optimization problem to provide valid post-selective Bayesian inference. 
The selection procedures are posed as data-queries that solve a randomized version of a convex learning program which have the advantage of preserving more left-over information for inference. 

We propose a randomization scheme under which the approximating optimization has \textit{separable} constraints that result in a partially \textit{separable} objective in \textit{lower} dimensions for many commonly used selective queries. We show that the proposed optimization gives a valid exponential rate of decay for the selection probability on a large deviation scale under a Gaussian randomization scheme. On the implementation side, we offer a primal-dual method to solve the optimization problem leading to an approximate posterior; this allows us to exploit the usual merits of a Bayesian machinery in both low and high dimensional regimes when the underlying signal is effectively sparse. We show that the adjusted estimates empirically demonstrate better frequentist properties in comparison to the unadjusted estimates based on the usual posterior, when applied to a wide range of constrained, convex data queries. 
\end{abstract}

\maketitle

%%%%%%%%%%%%%%%%%%%%%%%%%%%%%%%%%%%%%%%%%%%%%%%%%%%%%%%%%%%%%%%%%%%%%%%%%%%%%%%%%%%%%%%%%%%%%%%%%%%%%%%%%%%%%%%%%%%%%%%%%%%%%%%%%%%%%%%%%%%%%%%%%%%%%%%%%%%%%%%%%%%%%%%%%%%%%%%

%--------------------------------------------------
\section{Introduction}
%--------------------------------------------------
\label{introduction}

A line of works \cite{exact_lasso, optimal_inference, tibshirani2016exact, loftus_quadratic, yang2016selective, randomized_response} has established methodology for exact and asymptotic selection-adjusted inference that provide frequentist coverage guarantees in the regression framework. The driving motivation to adjust for selection is that analysts commonly conduct queries on a database in order to select inferential questions of interest about the population parameters. Inference after such interactions with the data lacks frequentist properties like target coverage when the same data set is used later for answering these very same questions. A Bayesian perspective on modeling the post-selective problem as a truncation is advocated in \cite{yekutieli2012adjusted} and extensions of the former work to the more general set-up of linear models are proposed in \cite{selective_bayesian}. These works propose the use of a fixed parameter view where the truncation is applied to the data exclusively conditional on the parameter. This alters the posterior distribution after selection unlike the usual Bayesian variable selection framework in \cite{mitchell1988bayesian, george1997approaches} where the posterior is known to display inadapativity to selection. 

More precisely, the truncated view point on inference is based on a \textit{selective posterior}, formed by a truncated likelihood in conjunction with a prior that allows an analyst to inject a priori information on parameters in a model after selection. Such an approach has the additional flexibility in allowing the analyst to fix a model based on a parametrization that can be guided by a selection procedure.
%The merits of the adjusted Bayesian approach are manifested as quantitatively better properties
% such as frequentist coverages and Bayesian FCR (the frequentist analog of which is introduced in \cite{benjamini2005false}) 
%under non-informative priors in these afore-mentioned works. 
Motivated by the conditional approach of modeling Bayesian inference, the current work focuses on developing concrete, scalable methods that will allow the analyst to exploit the full potent of a Bayesian machinery post a wide range of constrained-convex learning programs. The Bayesian problem is by no means a trivial extension of the existing frequentist methods as it requires a closed form expression for the normalizer of the truncated likelihood. We describe the computational difficulties in providing Bayesian inference in the truncated framework and the contributions of this work more formally after introducing the \textit{selective posterior}.

%--------------------------------------------------
\subsection{Selective posterior} 
\label{posterior}
%--------------------------------------------------
A selective posterior modeled along the conditional approach has two components - a truncated likelihood and a prior distribution on the parameters in the likelihood. The truncation is imposed by selection as the analyst is interested in providing inference for a target parameter only if he observed the associated selection event. A generative model that the analyst is willing to impose on data post selection, together with the truncation to all realized values that lead to an observed selection event determine the truncated likelihood. The prior allows him to inject information on the target from his existing knowledge. 

Formally, variable selection is based on an observed data vector $S$ and the selection event of observing an active set of variables $ \hat{E}(S)=E$ can be described as $\{s: \hat{E}(s)=E\}$, the set of realizations of $S$ that lead to $E$. It is only after selection that a model is defined, in this case, a generative Bayesian model with a likelihood parametrized by $\beta^*$, denoted as $f(s\lvert \beta^*)$ and a prior $\beta^*\sim\pi.$ The goal is to provide inference for a target determined by selection event $E$; that is, we infer about the target only if we observe $E$. This truncates the generative law of the data conditional on the parameter, resulting in the selection adjusted likelihood
\[f(s\lvert \beta^*)1_{\{\hat{E}(s)=E\}}/\mathbb{P}(\hat{E}(S)=E\lvert \beta^*).\] 
In conjunction with the prior $\pi(\beta^*)$, selective Bayesian inference about the data adaptive target is possible by sampling from 
\begin{equation}
\label{selective:posterior:0}
\pi_E(\beta^*\lvert s) \propto \pi(\beta^*)\cdot \cfrac{f(s\lvert \beta^*)}{\mathbb{P}(\hat{E}(S)=E\lvert \beta^*)}
\end{equation}
namely, the selective posterior. 
%We see various examples that show better coverage and risk properties of the adjusted estimates based on the truncated approach in Section \ref{experiments}. 
%For now, we focus on the computational challenge involved in Bayesian inference using \eqref{selective:posterior:0}. 

As is evident from above, the normalizer of the truncated likelihood is the probability of the selection event, computed as a function of the parameters in the generative density. 
%Contributing to the technical difficulties of sampling from the above posterior is the intractability of the selection probability of observing output $E$ from pre-inference queries. 
While sampling from the truncated likelihood in a frequentist regime does not require knowledge of the normalizer (treated as a constant), the normalizer that typically lacks a closed form expression does contribute to the selective posterior in a Bayesian paradigm. Implementing a sampling scheme then becomes impossible in the absence of an expression for the normalizer to the truncated law. \cite{selective_bayesian} identifies this technical hurdle and proposes an approximation to general affine selection probabilities which gives rise to a pseudo \textit{selective posterior}. Sampling from the \textit{selective posterior} necessitates computing the approximation, cast as an optimization problem for each draw of the sampler. The efficiency of any standard sampler thereby, hinges on the computational cost of solving the optimization objective associated with this approximation. In most cases, this can be very expensive and hard to scale with larger sample size and regression dimensions. 

We propose in this work a randomization scheme for commonly used selection queries and offer an approximating optimization under the same to facilitate sampling from an approximate selective posterior. The three major gains in the current work associated with the proposed optimization are
\begin{itemize}
\vspace{-0.5em}
\setlength\itemsep{-0.1em}
\item an objective with simpler constraints on the optimizing variables, as opposed to polyhedral constraints in \cite{exact_lasso}
\item a partially separable objective function with separability in the selective constraints
\item a reduction in dimensions of the optimization objective (an objective with smaller number of optimizing variables).
\end{itemize}
Typically, for popular constrained queries like marginal screening, Lasso, forward stepwise etc., the optimization solves an objective in $\min(d+|E|,p)$ dimensions with $d$ as the size of the observed data vector $S$, $p$ as the dimension of regression and $|E|$ as the size of active set. The key idea behind these reductions is an upper bound to the normalizer that capitalizes on the structure of an inversion map associated with the randomized selective query and a change of measure induced by the same, discussed in Section \ref{background}. 

The problem of analytically getting an approximation for the normalizer is similar to the goals of variational Bayesian approaches in \cite{minka2001family, hoffman2013stochastic} that use a known parametric distribution to obtain an approximation to an intractable posterior based on the KL-divergence between the two posteriors. We adopt a different approach here by approximating an intractable integral, the normalizer of the truncated law as a function of the parameters in the model; we show that this approximation gives an asymptotic large deviation behavior of the exact normalizer under Gaussian randomization schemes. 

These contributions allow a wider scope of applications of the truncated Bayesian approach to different generative models, randomization schemes and constrained selective queries. Such reductions become very useful in not just higher dimensions, but also, in providing inference after multiple selective queries. Below, we describe methods that demonstrate scalability of truncated Bayesian inference in both high and low dimensional regimes of inference. The effectiveness of the proposed methods is corroborated through Bayesian effect estimates with superior frequentist properties for data-mined variants in a real data set that investigates gene associations with local variants.

%--------------------------------------------------
\subsection{A motivating example}
\label{motivation}
%--------------------------------------------------
\label{example}
Before introducing our methods, we give an example that motivates the readers towards the inferential gains of the selective posterior over the more common inadaptive Bayesian approach. 
%can motivate readers towards the importance of the methods developed in the paper as opposed to an inadaptive Bayesian approach. 
Consider data $Y\in \real^n$ and a fixed predictor matrix $X \in \real^{n \times p}$ with columns scaled by $1/\sqrt{n}$ such that the response is generated as $Y = X \beta + \epsilon, \;\epsilon \sim \mathcal{N}(0, \sigma^2 I_n)$ given a $\beta \in \real^p$ and $\sigma^2=1$. 
%Note that the above is the true generative model for the data, but not part of the data analysis.
 An analyst decides to run Lasso on data $(Y, X)$ in order to choose $E$, a set of selected predictors.
 %perform inference on the population coefficients
%\[\Theta_E(\beta)= (X_E^T X_E)^{-1} X_E^T X\beta \in \real^{|E|}\]
%based on a model $Y\sim \mathcal{N}(X_E\beta_E, I_n)$ 
%where  using a Lasso query on the data. 
Not having access to the actual generative model, he assumes the screened model from Lasso as a plausible model on his data, that is
$Y\sim \mathcal{N}(X_E\beta_E, I_n)$ and a non-informative prior $\pi$ on the parameters $\beta_E$ in the selected model to offer Bayesian inference on $\beta_E$. Ignoring selection, he uses the unadjusted posterior on $\beta_E$
\begin{equation}
\label{unadjusted:posterior}
 \pi(\beta_E\lvert Y=y) \propto \pi(\beta_E) \cdot {\exp\left(-\|y- X_E\beta_E\|_2^2/2\sigma^2\right)}
\end{equation}
to report credible intervals and the posterior mean as inference for target $\beta_E$. 

We compare the estimates from the above approach of the analyst to truncated inference post a randomized version of the Lasso query. We give inference on $\beta_E$ using the same selected model and non-informative prior as the analyst where $E$ is the output from
\begin{equation}
\label{lasso:random}
\text{minimize}_{\beta}\frac{1}{2}\|y-X\beta\|_2^2-\omega^T\beta + \lambda \|\beta\|_1 + \frac{\epsilon}{2}\|\beta\|_2^2.
\end{equation}
Randomization enters the objective as $\omega^T\beta$, perturbing selection that is otherwise based only on $y$; the above randomized version of Lasso has been proposed in \cite{harris2016selective}. The objective has a small added ridge penalty $\epsilon = {1}/{\sqrt{n}}$ for existence of a solution and tuning parameter is set as $\lambda = \mathbb{E}[\|X^T \psi\|_{\infty}]$ as proposed in \cite{negahban2009unified} where $\Psi \sim \mathcal{N}(0, I)$. On a high level, our method of providing estimates in the truncated regime involves approximating the intractable posterior truncated to the realizations $(y,\omega)$ that lead to the same selection event. We finally use a Langevin walk-based sampler to provide adjusted Bayesian inference based on the approximate posterior. 

%highlights the necessity of inference based on the truncated model against the usual Bayesian paradigm. 
To compare our methods against the traditional Bayesian inference, we conduct the below experiment with two different generative mechanisms, Model I is a frequentist model with no signal and Model II is a Bayesian model. Let
$X \in \real^{n \times p}, \;n=200, \;p=1000$ be a design matrix with independent Gaussian entries normalized to have column norm $1$. %Randomziation $\omega$ is an instance of $\Omega\sim \mathcal{N}(0,\tau^2 I_p)$ in \eqref{lasso:random}.
\begin{itemize}
\setlength\itemsep{0.1em}
\item Model I: Draw in each trial $Y \sim \mathcal{N}(0, I_n);\;n =200.$
\item Model II: Draw in each trial $\beta \sim \pi^{\text{TRUE}}$ and $Y\lvert X, \beta$ as follows
\begin{enumerate}[(1).]
\setlength\itemsep{-0.2em}
\item $\beta \sim \pi^{\text{TRUE}}(\beta) = \prod_{j=1}^{1000} \pi^{\text{TRUE}}_j(\beta_j)$
with \[\pi^{\text{TRUE}}_j(\beta_j) = 0.90 \cdot \frac{1}{2b_1}\exp\left(-|\beta_j|/b_1\right) + 0.10 \cdot \frac{1}{2b_2}\exp\left(-|\beta_j|/b_2\right);\]
$b_1$ and $b_2$ represent the variance parameters of the Laplace densities in the mixture prior. In the below experiment, we set $b_1 = 0.1$ and $b_2= 1.0$ to generate an effectively sparse vector $\beta \in \real^p,\; p=1000$.
\item $Y\lvert X, \beta = X\beta + \epsilon ; \; \epsilon \sim \mathcal{N}(0,I_n);\; n =200.$ 
\end{enumerate}
\end{itemize}
In each trial, set target and model of inference based on observed $E$ post Lasso query in \eqref{lasso:random} as described above; we compare estimates based on \eqref{unadjusted:posterior} in the untruncated regime against our method of inference that gives adjusted estimates. To conduct the randomized query \eqref{lasso:random}, we draw $\omega$ as an instance of $\Omega\sim \mathcal{N}(0,\tau^2 I_p)$ in every trial. The below table gives a comparison of coverage of the credible intervals and risk in the frequentist model and of Bayesian FCR and Bayes risk of the posterior mean in the Bayesian model after $50$ trials. The target coverage for the intervals is set at $90\%$. Bayesian FCR in \cite{yekutieli2012adjusted} is defined as $\mathbb{E}_{\beta, Y}(V/\max(R,1))$, where $V$ is the number of non-covering credible intervals and $R=|E|$ is the number of intervals constructed after selection.
%; this marginalizes over the parameter as opposed to the frequentist analog, introduced in \cite{benjamini2005false}. 
Consistent with coverage, we report the proportion of $|E|$ intervals covering the target in the Bayesian model in Table \ref{table:motivating:example:2} and call it CR. Unlike the non-randomized intervals in \cite{exact_lasso} that are known to grow very wide, the power inherited from randomization is reflected in shorter lengths of the adjusted intervals. The results clearly highlight the superior frequentist properties of our methods, both in terms of coverage of credible intervals and risk of posterior mean.

%The inferential parameters based on which we assess the credibility of Bayesian inference in both the regimes are 
%\begin{itemize}
%\item expected proportion of covering constructed intervals under the underlying Bayesian generative model;
%\item the Bayes risk of the posterior mean under the true generative model;
%\item the lengths of the credible intervals, an indicator of inferential power.
%\end{itemize}
%The first criterion gives a measure of the Bayesian FCR (see the frequentist analog in \cite{benjamini2005false}); the second one gives an assessment of the performance of the posterior mean in both regimes and the third one reflects the power of our procedure. The intervals in \cite{exact_lasso} are known to grow very wide near the selective boundary due to loss of information for inference. Randomization as introduced in \cite{randomized_response} bounds the left over information away from $0$ which is reflected in the intervals based on our methods that are now much more comparable to the lengths of the unadjusted intervals.
%denote
%\begin{itemize}
%\item $\bf(TR)$ - inference based on a sampler targeting an approximate truncated posterior post a randomized Lasso query (discussed in Sections \ref{optimization} and \ref{sampler})
%\item $\bf(UR)$ - unadjusted Bayesian inference based on \eqref{unadjusted:posterior} post the outputs from a randomized query in \eqref{lasso:random}
%\item $\bf(UNR)$ - unadjusted Bayesian inference based on \eqref{unadjusted:posterior} post non-randomized query in \eqref{lasso:nonrandom}.
%\end{itemize}

\begin{center}
\vspace{-1.0em}
\small
\captionof{table}{Model I :  Coverage, Risk and length of intervals}
\label{table:motivating:example:1}
\bgroup
\def\arraystretch{1.}
\begin{tabular}{ |c|c|c|c|c| }
\hline
\bf{Method} & \bf{Coverage} & \bf{Risk} &  \bf{Lengths} \\
\hline
\hline
Truncated inference & $89.70\%$ & $1.81$ & $4.41$\\
\hline
Unadjusted inference & $51.38\%$ & $3.38$ & $3.34$\\
\hline
%Unadjusted inference post \eqref{lasso:nonrandom} & $46.45\%$ & $3.62$ & $3.41$\\
%\hline
\end{tabular}
\egroup
\end{center}

\begin{center}
\vspace{-.5em}
\small
\captionof{table}{Model II :  Bayesian CR, Bayes risk and length of intervals}
\label{table:motivating:example:2}
\bgroup
\def\arraystretch{1.}
\begin{tabular}{ |c|c|c|c|c| }
\hline
\bf{Method} & \bf{CR} & \bf{Bayes risk} &  \bf{Lengths} \\
\hline
\hline
Truncated inference & $90.99\%$ & $1.49$ & $4.49$\\
\hline
Unadjusted inference & $34.86\%$ & $4.28$ & $3.34$\\
\hline
%Unadjusted inference post \eqref{lasso:nonrandom} & $61.96\%$ & $2.83$ & $3.41$\\
%\hline
\end{tabular}
\egroup
\end{center}

The rest of the paper is organized as follows. Section \ref{background} outlines the truncated framework, giving the recipe for adjusted Bayesian inference using a selective posterior. Section \ref{optimization} lays out the backbone of the paper, the approximating optimization problem that we solve to sample from a tractable version of the selective posterior and provides a sampler that targets the approximate posterior. Section \ref{LDP} shows asymptotic validity of the finite sample bounds in Section \ref{optimization} to the otherwise unavailable normalizer for non-local sequences of parameter. Section \ref{sampler} lays out the optimization-based approach for popular selection queries. Section \ref{experiments} includes simulations that demonstrate the inferential gains associated with the truncated Bayesian methods in the current work over the unadjusted analog. Many of these examples bring to light the robustness of our methods to model mis-specifications. Section \ref{experiments} concludes with an application of our methods to provide adjusted Bayesian effect size estimates for local genetic variants (GTEx gene association data set) that have been data-mined as the strongest effects. 

%--------------------------------------------------
\section{A formal background}
%--------------------------------------------------
\label{background}

%In this section, we familiarize the readers with the \textit{selective posterior}, the main object of inferential interest in the paper. We give a description of the selective analysis conducted as randomized convex queries on the data and the inversion map associated with such a query. We outline the formal framework of inference, revisiting the \textit{selective posterior} post a randomized query. We conclude the section with the major contributions of the current work in the direction of making Bayesian computations in the truncated paradigm feasible.
\subsection{A randomized query and an inversion map} 
\label{randomized:query}
Selection events can be broadly viewed as outputs from queries on a data-base. In the context of variable selection, we are typically interested in the active set of coefficients obtained upon solving convex optimization problems. As a follow-up on the recent work on randomized inference in \cite{randomized_response, harris2016selective}, these queries are randomized versions of learning problems with a convex loss $\ell(S(X,y),.)$ and a convex penalty $\mathcal{P}_{\lambda}(.)$ with tuning parameter $\lambda$. Though the skepticism with randomization is that different instances of randomization can result in different selection outputs, we view it to be similar in spirit to the much practiced data-splitting where difference in outputs can result from various splits. Just as we can aggregate over the outputs from multiple splits of the data, we can similarly combine selections from multiple queries on the data-base as illustrated in \cite{markovic2016bootstrap}. Also, sharing similarity with the concept of reusable hold-out introduced in the field of differential privacy \cite{dwork2015preserving}, these forms of randomized inference come with the merit of higher statistical power during inference. For the Bayesian problem, randomization results in empirical improvements in the frequentist properties associated with the selective posterior, see \cite{selective_bayesian} for examples illustrating robustness of the randomized-credible intervals. The empirical results in Section \ref{experiments} of the current work corroborate these merits of a randomized Bayesian procedure, reflected in the coverage properties and shorter lengths of intervals. To add to these advantages, we leverage randomization to obtain significant computational reductions in solving an approximating optimization to sample from the selective posterior in the current work. The gains associated with randomized queries become clear after details in Section \ref{optimization}.

A randomized selective query taking a convex loss $\ell(S(X,y),.)$ and a convex penalty $\mathcal{P}_\lambda(.)$ as inputs, assumes the canonical form 
\begin{equation}
\label{canonical:randomized:program}
\hat{\beta}(s,\omega)= \argmin_{\beta} \ell(s;\beta) + \mathcal{P}_\lambda (\beta) -\omega^T \beta
\end{equation}
with data realization $s$ and randomization instance $\omega$, where $S\sim F$ independent of randomization $ \Omega \sim G$. 
The above algorithm has a linear term in randomization $\omega$, drawn from a distribution $G$ with a density $g$, fully supported on $\mathbb{R}^p$. This can be viewed as selection with a perturbed version of data, hence the term ``randomized" program. Queries of the above form are termed as objective perturbation in the privacy literature, see \cite{chaudhuri2009privacy, chaudhuri2011differentially}. Some randomized programs like the Lasso have an additional $\ell_2$ penalty term $ \frac{\epsilon}{2}\|\beta\|^2_2$ as in ridge regression in \cite{elastic_net} to enforce existence of a solution. The analyst has access to the output $E$, a function of $(s, \omega)$ from such a query in the inferential stage, typically the set of active coefficients along with their signs: see \cite{kac_rice, spacings, exact_lasso}. Selective inference seeks to overcome bias from having known the output of query prior to inference through the conditional approach.

 \cite{optimal_inference} presents a more natural analog of classical data-splitting in the form of data-carving, which advocates a random split of the data for selection, but allows the analyst to use the entire data for inference. A data-carved query that is performed on a randomly chosen split of the data is given by
\begin{equation}
\label{canonical:carved:program}
\hat{\beta}(s^{(1)},\omega)= \argmin_{\beta} \frac{1}{r}\ell(s^{(1)};\beta) + \mathcal{P}_\lambda(\beta).
\end{equation}
with $r$ as the fraction of data-samples used in selection and $s^{(1)}$ as a random split of the data-vector $s$.
\cite{markovic2016bootstrap} shows that the above selection can be cast as a randomized query of the form \eqref{canonical:randomized:program}. This can allow an analyst to collect new data and view prior selection on existing data as a split on the updated data-base. Hence, it can facilitate valid inference post already conducted exploratory analyses on existing data-bases while having extra power in comparison to an analysis on the new data set only. We discuss the data-carved version of Bayesian inference in more details later in Section \ref{sampler}.

%We emphasize on our interest in randomized selection algorithms at this point as it has added merits of better frequentist coverage properties in the current Bayesian methodology \cite{selective_bayesian} and increased statistical power \cite{randomized_response}. We also see computational reductions in solving the approximating optimization-``the key to getting a hand on the selective posterior" in the above randomized settings. A reduction is possible due to a change of measure formula, an idea used in \cite{zhou_montecarlo_lasso} and more recently, in \cite{harris2016selective} for sampling from a selective density. 
%In this work, we consider randomized queries of the form \eqref{canonical:randomized:program} that yield affine constraints on optimization variables.

The starting point of achieving computational reductions in approximating the selective posterior is an inversion map that characterizes the output from randomized queries. Such a map is obtained from the subgradient equation of \eqref{canonical:randomized:program}. The canonical selection event of observing an active set of coefficients $E$ with signs $z_E$ can be described in terms of the randomization $\omega$ and data instance $s$ using the inversion map. Denoting by $\hat{\beta}=(\hat\beta_E,0)$ the solution from a query in \eqref{canonical:randomized:program} with $\hat\beta_E$ as the active non-zero coefficients, the inversion map is given by
\begin{equation}
\omega = \partial \ell(s;(\hat\beta_E,0)) + \partial \mathcal{P}_\lambda((\hat\beta_E,0)).
\end{equation}
The above equation maps the randomization instance to realized data $S=s$ and  optimization variables $O= (\hat\beta_E, v_{-E})$ where $\hat\beta_E$ denotes the active coefficients and $v_{-E}$ represents the inactive sub-gradient corresponding to the inactive coordinates of $\partial \mathcal{P}_\lambda((\hat\beta_E,0))$. We denote the optimization variables corresponding to the active coordinates as $O_E$ and the ones corresponding to the inactive subgradient variables as $O_{-E}$ from now on, referring them as active and inactive optimization variables respectively.
Post an affine randomized selection event, the canonical inversion map that is the basis of a new measure takes the form
\begin{equation}
\label{randomization:map}
\omega(s,o) = Ds + P o + q
\end{equation}
with $s$ and $o$ representing data and optimization variables respectively and $D,P, q$ are fixed.

%A wide range of popular algorithms with different losses and penalties have the above canonical structure for selection event $(E, z_E)$, examples being marginal screening, Lasso and logistic Lasso at fixed $\lambda$, sequential methods like fixed steps of forward stepwise, top $K$ correlated, graphical models like neighborhood selection etc. 
The scope of randomized queries is quite broad in nature allowing for even discrete versions of randomizations like carving. In practice, the analyst may use a union of outputs $\bar{E} = E_1 \cup E_2 \cdots \cup E_k$ or the final model $E_k$ when $E_k\subset E_{k-1}\cdots \subset E_1$ (or some reasonable combination) based on a sequence of outputs $(E_1, E_2,...,E_k)$ from multi stagewise selective algorithms to determine a target and a generative mechanism in the inferential stage. We demonstrate the extension of our approach to multiple data queries using a combination of these model selection methods in \ref{multi:selection} under Section \ref{optimization}. This finds similarity in the method of approximately estimating expectations after allowing an analyst to repeatedly query a database in \cite{dwork2015preserving}.

\subsection{A Bayesian inferential scheme using inversion map}
\label{framework}
The ingredients for selective Bayesian inference are the same as the usual one, a prior and a likelihood, except that we replace the usual likelihood with a truncated one. 
%The generative likelihood conditioned on the selection event forms the truncated likelihood; the selection transforms the data generating mechanism through truncation. 
To describe our inferential framework, we assume a model $f(.\lvert \beta^*)$, parametrized by $\beta^*$ post-selection on data $S\in \real^d$ and fix a target denoted as $\Theta_E(\beta^*)$. In the linear model settings with a fixed design matrix $X\in \mathbb{R}^{n\times p}$, $f(.\lvert \beta^*)$ might correspond to a family of models $\mathcal{F}_{\beta^*} = \{\mathcal{N}(X^*\beta^*, \sigma^2 I): \beta^*\in \mathbb{R}^k\}$
for a known $\sigma$ with $X^*=\begin{bmatrix} X_{i_1} & \cdots X_{i_k}\end{bmatrix}$, for a set of indices $\{i_1,i_2,\cdots, i_k\}\subset \{1,2,\cdots, p\}.$ We emphasize here that we do not have an idea about $\beta^*$ before we run a selection mechanism like the Lasso. There are some settings where the parameterization exists before selection and does not change, example being the saturated model of \cite{exact_lasso}. Typically, we are running Lasso in order to find something that might be an interesting parameterization. 

A common target of inference post selection of an active set $E$ is the usual population coefficient corresponding to ordinary least squares on the selected model $E$, that is
\[\Theta_E(\beta^*) = (X_E^T X_E)^{-1} X_E^T\mathbb{E}[Y\lvert \beta^*] = (X_E^T X_E)^{-1} X_E^T X^*\beta^*.\]
With a random design matrix, the target of inference can be described as 
\[\Theta_E(\beta^*) = (\mathbb{E}[X_E^T X_E])^{-1} \mathbb{E}[X_E^T Y\lvert \beta^*],\]
with the generative family of models parametrized as $\{f(.\lvert \beta^*): \mathbb{E}_f[Y\lvert X] = X^*\beta^*\}$.
\begin{remark}
\emph{\it Generative models and targets:}
The selected model described in \cite{optimal_inference} corresponds to parametrization $\beta^*= \beta_E\in \mathbb{R}^{|E|}$ with $E$ being the observed active set and the saturated model corresponds to a parametrization $\beta^* =\mu \in \mathbb{R}^n$. The corresponding $X^*$'s in the two models are $X_E$ and the identity matrix $I_n$ respectively. Of course, other models are possible. The analyst can allow selection to guide him to a target and a generative model, though these choices do not necessarily have to agree with the observed selected set $E$. The methods of inference described here are flexible to allow him to use expert opinion on a plausible generative model parametrized by $\bar{E}$, and a possibly more interesting target $\Theta_{\bar{E}}(\beta^*)$, where $\bar{E}$ is determined through $E$. To be able to highlight the flexibility of our method to various parametrizations of the mean, we use the general notation $\beta^*$ to denote the parameters underlying the Bayesian model assumed post selection.
\end{remark}
\begin{remark} 
\emph{\it Prior on variance parameter:} The variance $\sigma$ in the generative likelihood can be modeled in a Bayesian paradigm by putting a joint prior on $(\beta^*, \sigma)$. We do not delve into details of incorporating a Bayesian model on the variance in the current draft; hereafter, we stick to a fixed variance setting.
\end{remark}

Using a change of measure based on the inversion map in \eqref{randomization:map}, the joint truncated density at $(s,o)$ corresponding to a generative model $f$ on data in $\mathbb{R}^d$ with parameters $\beta^*$ and randomization density $g\in \mathbb{R}^p$  decouples as
\begin{equation}
\label{selective:density}
h_E(s, o) \propto  |J|  \cdot f(s\lvert \beta^*)\cdot g(Ds + P o + q)
\end{equation}
with support 
$$\mathcal{R}= \{(s,o): \hat{E}(s,o) = E, z_{\hat{E}}(s,o)= z_E\}=  \mathbb{R}^d \times \mathcal{R}_O.$$
$|J|$ is a Jacobian reflecting the change of measure, a constant for affine inversion maps as in \eqref{randomization:map}. The support is unrestricted on data $s$ and constrained to $\mathcal{R}_O\subset \mathbb{R}^p$, representing constraints on optimization variables, imposed by selection output. \cite{harris2016selective} advocates this new measure in order to enable a frequentist to sample from a density with a fairly simple support region as opposed to more general affine constraints on data and randomization. 

Coming back to a Bayesian setting, the selective posterior for generative parameters $\beta^*$ given data $S$ when $\beta^*\sim \pi(\beta^*)$  is formed by appending the marginal selective density of $S$ to the prior $\pi(.)$. The truncated marginal of $S$ given parameters $\beta^*$ is obtained by marginalizing over $O$ in the joint density \eqref{selective:density}. The selective posterior is thus, given by 
\begin{equation}
\label{selective:posterior}
 \pi_E(\beta^*\lvert S) \propto \pi(\beta^*) \cdot \cfrac{f(s\lvert \beta^*)}{\mathbb{P}((S,O)\in \mathcal{R}\lvert \beta^*)}.
\end{equation}
The above posterior is however intractable as the normalizer 
\[{\mathbb{P}((S,O) \in \mathcal{R}\lvert \beta^*) =  \int_{\mathcal{R}}   |J|  \cdot  f(s\lvert \beta^*)g(D s + P o + q)dods}\]
 has no exact closed form expression. The problem reduces to computing the normalizer $\mathbb{P}((S,O) \in \mathcal{R}\lvert \beta^*)$; we focus on this in Section \ref{optimization}.%We instead propose using a pseudo posterior $\tilde{\pi}_E(.\lvert S)$ in this work based on a tractable and scalable approximation to the otherwise unavailable normalizer.
 
 %--------------------------------------------------
\section{An approximating optimization}
%--------------------------------------------------
\label{optimization}

%This section presents an approximating optimization problem based on a change of measure induced by the inversion map in \eqref{randomization:map}. For randomizations that are independent in all coordinates, we present a modified approximation with reduced optimization objective; such a reduction is possible due to decoupling of the randomization components and an exact expression for the selection probability on the inactive subgradient variables. We present the duals of the approximating optimizations which can efficiently scale in some situations than the corresponding primal problem. Finally, we extend our approximation to multiple randomized queries on the data that can account for multi-level or many selections to provide adjusted inference on a suitably determined target. 

\subsection{Approximate normalizer based on inversion map}
\label{approx:normalizer}

Using the inversion map that defines the selection output from a query in \eqref{randomization:map}, we derive an approximating optimization with a constrained objective in $d+p$ dimensions that bounds from above the log normalizer. 
%This serves as a starting point of obtaining an approximation cast as an optimization problem.
%The Chernoff bound proposed in the afore mentioned work for selection probability $\mathbb{P}(Z \in \mathcal{R}\lvert \beta^*)$ based on the MGF $\Delta_Z$ of a random variable $Z$ for a compact, convex $\mathcal{R}$ is given by
%\begin{equation}
%\label{Chernoff:bound}
%\log \mathbb{P}(Z\in \mathcal{R}\lvert \beta^*) \leq -\inf\limits_{z\in \mathcal{R}}\Big\{\sup\limits_{\alpha}\{ z^T \alpha  -\log\Delta_Z(\alpha)\}\Big\} 
%&\;\;\;\;\;\;\;\;\;\;\;\;\;\;\;+ b_\mathcal{R}(z)\Big\}.\numberthis\label{sel:prob:approx}
%\end{align}
%\end{equation}
%A Chernoff bound for selection probability $\mathbb{P}(Z \in \mathcal{R}\lvert \beta^*)$ based on the MGF $\Delta_Z$ of a random variable $Z$ for a compact, convex $\mathcal{R}$, derived in the afore-mentioned work, is given by
%\begin{equation}
%\label{Chernoff:bound}
%\log \mathbb{P}(Z\in \mathcal{R}\lvert \beta^*) \leq -\inf\limits_{z\in \mathcal{R}}\Big\{\sup\limits_{\alpha}\{ z^T \alpha  -\log\Delta_Z(\alpha)\}\Big\} 
%\end{equation}
We state below the first theorem of this paper that gives rise to an upper bound on the volume of a convex and compact selection region $\mathcal{R}$ with respect to the joint density of data and optimization variables. It involves computating the log-MGF of the augmented vector of data and optimization variables with respect to a transformed measure induced by the inversion map in \eqref{randomization:map}. 
\begin{theorem}
\label{approximate:prob} Denoting $\Lambda_f^*(.\lvert \beta^*) $ as the convex conjugate of the log-MGF $\Lambda_f(.\lvert \beta^*)$ of data vector $S \in \real^d$ and $ \Lambda_g^*(.)$ as the conjugate of the log-MGF $\Lambda_g$ of randomization $\Omega\in \real^p$, a Chernoff upper bound to the exact selection probability $\log\mathbb{P}((S,O) \in \mathcal{R}\lvert \beta^*)$ for convex, compact $\mathcal{R} \subset \real^d \times \real^p$ under the canonical inversion map in \eqref{randomization:map}
is given by
\begin{equation}
\label{chernoff:version}
\begin{aligned}
-\inf_{s,o \in \mathcal{R}} \Big\{\Lambda_f^*(s\lvert \beta^*)+  \Lambda_g^*(Ds + P o + q)\Big\}
\end{aligned}
\end{equation}
\end{theorem}
We prove the above in Appendix \ref{A:1}. While the above upper bound does hold for compact selection regions, the canonical selective constraints lead to a selection region of the form 
\[\mathcal{R} = \mathcal{R}_S \times \mathcal{R}_O \text{ with } \mathcal{R}_S =\mathbb{R}^d\] 
and $\mathcal{R}_O$ is typically tensor of orthants and cubes; this lacks compactness. The upper bound derived in \ref{approximate:prob} can still be applied as an approximation as we can work with a sufficiently large compact and convex subset of $\mathcal{R}$ that has an almost $1$-measure under prior $\pi$. A smooth version of \eqref{chernoff:version} is seen to lead to better frequentist properties in \cite{selective_bayesian} in the non-randomized settings; in the current work, we opt for \eqref{barrier:version} to solve a smooth objective in place of a constrained optimization.

The bound-based approximation above is given by
\begin{equation}
\label{chernoff:min}
-\inf_{s\in \mathbb{R}^d,o \in \mathbb{R}^p} \Big\{\Lambda_f^*(s\lvert \beta^*)+  \Lambda_g^*(Ds + P o + q)+\chi_{\mathcal{R}_O}(o)\Big\}
\end{equation}
with $\chi_{\mathcal{R}_O}(.)=-\log 1_{\mathcal{R}_O}$. In particular,  $\chi_{\mathcal{R}_O}(.)$ can be interpreted as a function with a uniformly $0$ penalty within the selection region. An improved approximation to the selection probability can be obtained by smoothing the discrete penalty $\chi_{\mathcal{R}_O}(.)$ in the bound with a barrier penalty $b_{\mathcal{R}_O}(.)$, which imposes a continuously decaying penalty as distance from the selective boundary increases. This leads to a smooth, unconstrained version of \eqref{chernoff:version} to approximate $\log\mathbb{P}((S,O) \in \mathcal{R}\lvert \beta^*)$ and is given by 
\begin{equation}
\label{barrier:version}
%\begin{align}
{ - \inf\limits_{s \in \mathbb{R}^{d},\;o \in \mathbb{R}^{p}} \Big\{\Lambda_f^*(s\lvert \beta^*) +  \Lambda_g^*(Ds + P o + q) + b_{\mathcal{R}_O}(o)\Big\}}%\end{align}
\end{equation}
using a \textit{barrier penalty} $b_{\mathcal{R}}(.)$ on affine constraints induced on the optimization variables.
%A particular choice of barrier function, called the \textit{soft max} barrier for polyhedral constraints $\mathcal{R} =\{a_i^T z \leq b_i; \;i =1,2,...,m\}$ is given by
%\[ b_\mathcal{R}(z)=\begin{cases} 
%      \sum_{i=1}^m \log\left(1 + \cfrac{1}{b_i - a_i^T z}\right) & \text {if }z\in \mathcal{R} \\
%      \infty & \text{ otherwise}.
%      \end{cases}
%\]
%We use the above barrier function in the implementations of the current work, except that the constraints in our case simplify to typically sign and cube constraints as in \cite{harris2016selective} instead of the complicated affine constraints like in \cite{exact_lasso}. 
The gain with \eqref{barrier:version} in comparison to the prior work is a much easier objective function as the canonical constraints on the optimization variables simplify to sign and cube constraints as in \cite{harris2016selective} instead of the complicated affine constraints as in \cite{exact_lasso}. We can further benefit from separability and achieve more reductions from such an approximation under certain randomizations, as seen later in \eqref{approx:decomposed}. 

The unconstrained optimization given by \eqref{barrier:version} in $d+p$ dimensions can be used to approximate selection probabilities under any randomization with a log-MGF $\Lambda_g$, that is independent of the data vector. In particular, we can use the optimization for inference post data carved queries of the form \eqref{canonical:carved:program}. Randomization in such queries takes the form of the gradient of difference of losses 
\[\omega = \partial \ell(s; (\hat\beta_E,0))-\frac{1}{r}\partial \ell(s^{(1)}; (\hat\beta_E,0)).\] 
and is asymptotically independent of the data vector for a Gaussian generative model and marginally an asymptotic centered Gaussian with a covariance $\Sigma_g$. Using the conjugate of the log-MGF of a Gaussian density, we obtain a tractable pseudo posterior. We illustrate inference based on the approximate selective posterior post selection on a random fraction of the data in Section \ref{sampler}.
%For data-carved queries, \cite{markovic2016bootstrap} show that the randomization is asymptotically independent of the data vector. This allows us to implement the above optimization to approximate the \textit{selective posterior} with selection on a random split of the data.

\subsection{Reduction in optimization}
\label{reduction}

Under randomizations with a density supported on $\mathbb{R}^p$ that are independent in all $p$-component coordinates, we present an approximation that is based on smoothing a modified upper bound. For most common queries, it involves an optimization objective in $d+|E|$ dimensions, where $|E| \leq p$ is the size of the active set from the selective query. Note that the optimization in \eqref{barrier:version} involves $d+p$ optimizing variables. With the reduction in dimensions of the optimization, we make a significant improvement in scalability of our methods in high dimensional sparse problems, when $|E|\ll p$. Such a reduction is possible due to 
\begin{itemize}
\item decoupling of randomization density under independence 
\item the structure of the canonical inversion map in \eqref{decomp:rand} that allows an exact and easy calculation of the volume of the inactive selection region with respect to the density of $O_{-E}$. 
\end{itemize}

Before proceeding further, consider a break-up of the canonical randomization map into $E$ active and $p-|E|$ inactive coordinates.  Such a decomposition takes the form
\begin{equation}
\label{decomp:rand}
{\omega(s,o) = Ds + P o + q =\begin{pmatrix} D_E s + P_E o_E + q_E \\ D_{-E}s + P_{-E} o_{E} + o_{-E} + q_{-E}\end{pmatrix}}
%&=\begin{pmatrix} D_E s + P_E o_E + q_E \\ D_{-E}s + P_{-E} o_{E} + o_{-E} + q_{-E}\end{pmatrix}, \numberthis \label{randomization:map:decomp}
\end{equation}
where $o_E$ denotes the active coefficients and $o_{-E}$ represents the inactive subgradient. The inversion map has such a structure in most commonly used queries like the Lasso, forward stepwise, thresholding etc. as we see later in Section \ref{sampler}. The density $g$ under a component-wise independent randomization scheme decouples into the active and inactive coordinates as
 \[{g(\omega) = g_E(\omega_E) \cdot g_{-E}(\omega_{-E})= g_E(\omega_E) \cdot \Pi_{j}g_{j,-E}(\omega_{j,-E})}.\] 
The constraints on $(o_E, o_{-E})$ for the canonical map are also separable and particularly, the inactive constraints are separable in each coordinate. The selection region induced by the selective constraints can thus, be denoted by
 \begin{equation}
 \label{separable:cons}
\mathcal{R}_O = \mathcal{R}_E \times \mathcal{R}_{-E}= \mathcal{R}_E\times \prod_{j} \mathcal{R}_{j,-E}
\end{equation}
where $\mathcal{R}_E$ represents the active constraint region, $\mathcal{R}_{-E}$ the inactive region and $\mathcal{R}_{j,-E}$, each component inactive constraint. The below theorem uses this separability in constraints and independence to obtain an upper bound on the logarithm of the normalizer of the truncated law. It involves computing the exact probability of the inactive subgradient variables lying in the selection region 
$\mathcal{R}_{-E} = \prod_j \mathcal{R}_{j,-E}$
as a function of realizations of the active optimization variable $o_E$ and data $s$.
\begin{theorem}
\label{bound:modified}
Under a randomization scheme composed of $p$ independent components $\Omega=(\Omega_1,\cdots,\Omega_p)$ and a selective query of the form \eqref{decomp:rand} yielding a compact and convex selection region 
\[\mathcal{R} = \mathcal{R}_S \times \mathcal{R}_O; \text{ where } \mathcal{R}_S \subset\mathbb{R}^d,\;  \mathcal{R}_O\subset \real^p\]
and $\mathcal{R}_O$ takes the form \eqref{separable:cons}, an upper bound for $\log \hat{\mathbb{P}}((S,O) \in \mathcal{R}\lvert \beta^*)$ for a compact, convex selection region $\mathcal{R}$
is given by
\begin{equation*}
 -\inf\limits_{s\in \mathcal{R}_S, o_E \in \mathcal{R}_E}\Big\{\Lambda_f^*(s\lvert \beta^*) +  \Lambda_{g_E}^*(D_E s +P_E o_E+ q_E)-\log\mathcal{B}(o_E;s) \Big\}
\end{equation*}
%\begin{equation}
%\label{approx:decomposed}
%\begin{aligned}
% \Scale[0.95]{-\inf\limits_{s \in \mathbb{R}^{n},\;o_E \in \mathbb{R}^{|E|}}\Big\{  \Lambda_f^*(s\lvert \beta^*) +  \Lambda_{g_E}^*(D_E s +P_E o_E+ q_E)+\mathcal{B}(o_E;s) + b_{\mathcal{R}_E}(o_E)\Big\} }
%\end{aligned}
%\end{equation}
with
\begin{equation*}
\label{approx:decomposed:barrier}
\begin{aligned}
%\mathcal{B}(o_E;s) &= \inf_{o_{-E} \in \mathbb{R}^{p-|E|}} \Big\{\Lambda_{g_{-E}}^*(D_{-E}s + P_{-E} o_{E} +\\
%&\;\;\;\;\;\;\;\;\;\;\;\;\;\;\;\;\;\;\;\; q_{-E}+ o_{-E} )+ b_{\mathcal{R}_{-E}}(o_{-E})\Big\},
\mathcal{B}(o_E;& s) = \prod_{j=1}^{p-|E|}\int_{\mathcal{R}_{j,-E}}g_{j,-E}(o_{j,-E} +D_{j,-E}s + P_{j,-E} o_{E}+ q_{j,-E})do_{j,-E}
\end{aligned}
\end{equation*}
where $D_{j,-E}$, $P_{j,-E}$ and $q_{j,-E}$ denote the $j$-th rows of the matrices $D_{-E}, P_{-E}$ and $j$-th component of vector $q_{-E}$ in \eqref{decomp:rand} respectively.
%with $\Lambda^*_{g_E}$ as log-MGF corresponding to $g_E$, the active density and similarly, $\Lambda^*_{g_{j,-E}}$ as log-MGF corresponding to $g_{j,-E}$, the inactive density in component $j$.
\end{theorem}

A proof of the bound is done in Appendix \ref{A:1}. A heuristic\footnote{exact is possible if selection region were compact; we still can apply the approximation with a large enough compact subset of the selection region with almost mass $1$ under the prior} minimax argument together with smoothing of constraints by a barrier penalty yields a reduced analog of \eqref{barrier:version} for canonical selective queries in the paper. An approximating optimization with a barrier penalty on the active constraints denoted as $b_{\mathcal{R}_E}(.)$ can be written as
%\begin{equation}
\begin{align}
\log \hat{\mathbb{P}}((S,O) \in \mathcal{R}\lvert \beta^*) &=-\inf\limits_{s \in \mathbb{R}^{d},\;o_E \in \mathbb{R}^{|E|}}\Big\{  \Lambda_f^*(s\lvert \beta^*) +  \Lambda_{g_E}^*(D_E s +P_E o_E+ q_E)\nonumber \\
 &\;\;\;\;\;\;\;\;\;\;\;\;\;\;\;\;\;\;\;\;\;\;\;\;\;\;\;\;\;\;\;\;-\log\mathcal{B}(o_E;s) + b_{\mathcal{R}_E}(o_E)\Big\}\label{approx:decomposed}
 \end{align}
%\end{equation}
with $\mathcal{B}(o_E; s)$ as defined in Theorem \ref{bound:modified}. 

Expression \eqref{approx:decomposed} yields an approximating optimization in $\mathbb{R}^{d+|E|}$ with a barrier function on the sign constraints of the active optimization variables in $\mathbb{R}^{|E|}$. We use the fact that the volume of the inactive selection region $\mathcal{B}(o_E;s)$ can be calculated exactly and easily as $p-|E|$ simple, univariate integrals over intervals $\mathcal{R}_{j,-E}\subset \real$. For example, for a centered Gaussian randomization with covariance matrix $\tau^2 I_p$ and the canonical cube constraints on the inactive subgradient variables $O_{-E}$ taking the form
$\mathcal{R}_{j,-E}=\{o_{j,-E} : |o_{j,-E}|\leq \lambda\},$
a closed form expression for the logarithm of the volume of the inactive cube region is
\[\log\mathcal{B}(o_E;s) = \sum\limits_{j=1}^{p-|E|} \log \left\{\Phi\left(\frac{\lambda+ \alpha(o_E; s)_j}{\tau}\right) - \Phi\left(\frac{-\lambda+ \alpha(o_E; s)_j}{\tau}\right)\right\}.\]
Here $\alpha(o_E; s)_j$ denotes the $j$-th coordinate of $\alpha(o_E; s) \in \mathbb{R}^{p-|E|}$, the Gaussian mean of $O_{j,-E}$ given $O_E = o_E, S=s$. Marginalizing over the inactive optimization variables results in a significant reduction in dimensions of optimization from the objective in \eqref{barrier:version}. Similar exact calculations of univariate probabilities of lying within an interval are easily available for other heavier tailed randomizations like the Laplace, Logistic etc. used in implementations in \cite{markovic2016bootstrap}.

\subsection{Dual problem: low dimensional regime}
\label{dual}
While solving the pseudo selective posterior using the above optimization as a surrogate to the normalizer is scalable for high dimensional problems, when $p\gg d+|E|$, it is not very ideal in the low dimensional regime with a large sample size, when $d+|E|\gg p$. Further, the optimization in \eqref{approx:decomposed} requires knowledge of the conjugates of the log-MGFs of the densities of the data and randomization. The dual problem yields an optimization objective in $\mathbb{R}^{p}$ and hence, renders a scalable version of the optimization in the low dimensional paradigm. The other distinction from the optimization posed in the primal form is that the dual is based on simply the log-MGFs corresponding to the distributions of data and randomization. In the low dimensional situation or when we do not have closed forms for the conjugates of the log-MGFs of the generative model, we can solve for the dual of the optimization problem instead.
\begin{theorem}
\label{dual:opt:appprox}
Denoting $\Lambda_f(.\lvert \beta^*) $ as the log-MGF of data generative density $f$ and $ \Lambda_g(.)$ as the log-MGF of randomization $\Omega$, the dual to the optimization approximating the selection probability $\log {\mathbb{P}}((S,O) \in \mathcal{R}\lvert \beta^*) $ in \eqref{barrier:version} is given by
%\begin{equation}
%\label{dual:chernoff}
%\begin{aligned}
% \inf_{u \in \mathbb{R}^p} \Big\{ \Lambda_f( D^T u\lvert \beta^*)+  \Lambda_g(-u)+ \chi^*_{\mathcal{R}_O}(P^T u)\Big\}
% \end{aligned}
%\end{equation}
%where $\chi^*_{\mathcal{R}}(.)$ is the conjugate of characteristic function of $\mathcal{R}$.
%where $\{x\}_j$ denotes the $j$-th coordinate of vector $x$.
%The dual of smoothened optimization in \eqref{barrier:version} is given by
\begin{equation}
\label{dual:barrier}
\begin{aligned}
\inf_{u \in \mathbb{R}^p} \Big\{ \Lambda_f(D^T u\lvert \beta^*)+ \Lambda_g(-u)+ b_{\mathcal{R}_O}^*(P^T u) + u^T q\Big\}
\end{aligned}
\end{equation}
where $b_{\mathcal{R}_O}^*$ is conjugate of the barrier function $b_{\mathcal{R}_O}(.)$ and $D,P, q$ are coefficients of linear terms of map \eqref{randomization:map}.
\end{theorem}

See proof in the appendix \ref{A:1}. A point to note is that dual formulation of the approximating optimization involves computing the conjugate of the barrier penalty function on the optimization variables. Since the constraints on the active and inactive optimization problems are separable, this involves solving conjugates of $|E|$ and $p-|E|$ univariate functions that correspond to the active and inactive constraints respectively. That is the conjugate barrier takes the additive form
\[b_{\mathcal{R}_O}^*(P^T u) = b_{\mathcal{R}_E}^*(P_E^T u) + b_{\mathcal{R}_{-E}}^*(P_{-E}^T u).\]
Details of the computation of the conjugates of the barrier functions used in our implementations are given in Appendix \ref{dual:details}.

\begin{remark} 
The dual of the constrained Chernoff-based optimization in \eqref{chernoff:min} for the canonical constraint region 
\begin{equation*}
\begin{aligned}
\mathcal{R}_O &= \prod_{j=1}^{|E|} \mathcal{R}_{j,E} \times \prod_{j=1}^{p-|E|}\mathcal{R}_{j-E}\\
& = \prod_{j=1}^{|E|} \{o_{j,E}: \text{diag}(s_{j,E}) o_{j,E}>0\} \times  \prod_{j=1}^{p-|E|}\{o_{j,-E}: |o_{j,-E}|\leq \lambda\}
\end{aligned}
\end{equation*}
 is given by
\begin{equation}
\label{dual:chernoff:can}
\begin{aligned}
 \inf_{u \in \mathbb{R}^p:\text{diag}(s_E)P_E ^T u < 0} \Big\{  \Lambda_f(D^T u\lvert \beta^*)+ \Lambda_g(-u)+ \lambda\sum_{j=1}^{p-|E|} |P_{j,-E}^T u | + u^T q\Big\}
\end{aligned}
\end{equation}
where $P_{j,-E}^T$ denotes the $j$-th row of matrix transpose of $P_{-E}$. This is by observing that the convex conjugate of the characteristic function $\chi_{\mathcal{R}_{j,E}}(.)$ representing the sign constraints on the active optimization variables is
\[\chi_{\mathcal{R}_{j,E}}^*(P_{j,E} ^Tu)= \begin{cases} 
      0 &\text{ if } s_{j,E}P_{j,E}^T u <0 \\
      \infty & \text{ otherwise}.
   \end{cases}
\]
and that for the cube constraints $\chi_{\mathcal{R}_{j,-E}}(.)$ on the inactive subgradients is
\[ \chi_{\mathcal{R}_{j,-E}}^*(P_{j,-E}^T u)= \lambda |P_{j,-E}^T u|.\]
\end{remark}

\subsection{Marginalizing over multiple selections}
\label{multi:selection}

The optimization problem described above is aimed to approximate the selection probability of an event based on a single randomized data query of the form \eqref{randomization:map}. It is however, common practice to apply stages of screening or query the data base multiple times to arrive at a selected set. An example might be laboratory A performing an initial scan of thousands of potential predictors to select a pool that passes a suitably chosen thresholding criterion and laboratory B conducting another screening of predictors. The analyst is interested in combining both screening results to guide her to inference on the same data set that has been analyzed by the two laboratories.

The approximation presented in \eqref{barrier:version}, \eqref{approx:decomposed} and \eqref{dual:barrier} can be marginalized over multiple randomizations from multiple stages and hence, be extended to multi-stage selective algorithms. The next Lemma renders an approximation to the normalizer for a $K$-stage randomized selection with query in each stage corresponding to an inversion map
\begin{equation*}
\label{multi:randomization:decomp}
\omega_k = D_k s + P_k o_k + q_k \text{ for } k =1,2,..,K
\end{equation*}
%\nocite{langley00}
with $o_k$ being the optimization variables for the randomized program in stage $k$. The selection region, determined by constraints on optimization variables $o_K$ at each stage, separable in the active and inactive coordinates as before, is given by
\begin{equation*}
\mathcal{R}_{(O_1,\cdots,O_k)} = \Pi_{i=1}^{K} \mathcal{R}_{O_k}.
\end{equation*}
Again, denote $\mathcal{R} = \mathcal{R}_S \times \mathcal{R}_{(O_1,\cdots,O_k)}$. Typically $\mathcal{R}_S = \real^d$, the unconstrained data-space augmented with the constrained region on optimization variables from each query.
\begin{lemma}
\label{approximate:prob:stagewise} Under $K$ randomizations with $\Omega_k \stackrel{\text{ind}}{\sim} g_k(.)$ for $k=1,2,...,K$ and with $ \Lambda_{g_k}^*(.)$ as the conjugate of the log-MGF $\Lambda_{g_k}$ of randomization $\Omega_k$, an upper bound to the logarithm of the exact selection probability $\log{\mathbb{P}}((S,O_1,\cdots, O_K) \in \mathcal{R}\lvert \beta^*)$ for a convex, compact $\mathcal{R}$
is given by
\begin{equation*}
\begin{aligned}
&- \inf\limits_{s \in \mathcal{R}_S, o_k \in\mathcal{R}_{O_k}, k=1,2,...,K}  \Big\{\Lambda_f^*(s\lvert \beta^*) +  \sum_{k=1}^{K} \Lambda_{g_k}^*(D_k s + P_k o_k + q_k)\Big\}.
\end{aligned}
\end{equation*}
%The dual formulation of above optimization is 
%\begin{equation*}
%\begin{aligned}
%&\inf\limits_{u_1,.., u_K} \Lambda_f\left(\sum_{k=1}^K D_k^T u_k\right) +\sum_{k=1}^K\Big\{ \Lambda_{g}(-u_k) +b_{\mathcal{R}_{O_k}}^*(P_k^Tu_k) + u_k^T q_k\Big\}%\\
%&\;\;\;\;\;\;\;\;\;\;\;\;\;\;\;\;\;\;\;\;\;\;\;\;\;\;\;\;\;\;\;\;\;\;\;\;\;\;\;\;\;\;\;\;\;\;\;\;\;\;\;b_{\mathcal{R}_{O_k}}^*(P_k^Tu_k)\Big\}
%\end{aligned}
%\end{equation*}
\end{lemma}
\begin{proof} 
The proof is easy to see as with independent randomizations in each stage of selection , we have $\log{\mathbb{P}}((S,O_1,\cdots, O_K) \in \mathcal{R}\lvert \beta^*)$ bounded from above by
\begin{align}
\Scale[0.93]{-\inf\limits_{s\in \mathcal{R}_S,o_k \in \mathcal{R}_{O_k},k=1,2,..,K}\Big\{\alpha^T s + \sum_{k=1}^K \alpha_k^T o_k\Big\} - \log\mathbb{E}(\exp(\alpha^T S + \sum_{k=1}^K\alpha_k^T O_k)\lvert \beta^*).} \nonumber 
\end{align}
An optimization over $\alpha \in \real$ and $\alpha_k \in \real, \{k=1,2,...,K\}$ and a minimax equality gives the bound
\[\Scale[0.95]{-\inf\limits_{s\in \mathcal{R}_S,o_k \in \mathcal{R}_{O_k},k=1,,..,K}\sup\limits_{\alpha, \alpha_k} \Big\{\alpha^T s + \sum_{k=1}^K \alpha_k^T o_k - \log\mathbb{E}(\exp(\alpha^T S + \sum_{k=1}^K\alpha_k^T O_k)\lvert \beta^*)\Big\}}.\]
A similar computation of the log-MGF of the augmented vector $(S,O_1,\cdots, O_K)$ as in the proof of Theorem \ref{approximate:prob} based on the change of variables facilitated by the inversion maps in \eqref{multi:randomization:decomp} completes the proof.
\end{proof}

The smooth analog of the constrained optimization in Lemma \ref{approximate:prob:stagewise} is given by
\begin{equation*}
- \inf\limits_{s \in \mathbb{R}^{d}, o_k \in\mathcal{R}_{O_k}, k=1,2,...,K}  \Big\{\Lambda_f^*(s\lvert \beta^*) +  \sum_{k=1}^{K} \Lambda_{g_k}^*(D_k s + P_k o_k + q_k) + b_{\mathcal{R}_{O_k}}(o_k)\Big\}.
\end{equation*}
The dual formulation of this approximation, optimizing over dual variables $u_k ;\; k=1,2,\cdots,K$ is given by
\begin{equation}
\label{dual:multiple}
\begin{aligned}
&\inf\limits_{u_1,.., u_K} \Lambda_f\left(\sum_{k=1}^K D_k^T u_k\right) +\sum_{k=1}^K\Big\{ \Lambda_{g}(-u_k) +b_{\mathcal{R}_{O_k}}^*(P_k^Tu_k) + u_k^T q_k\Big\}.
\end{aligned}
\end{equation}
\begin{remark}
\emph{Cost of optimization:} The optimization in Lemma \ref{approximate:prob:stagewise}, decomposed into active and separable inactive problems can be solved in its primal form in effectively $d+ \sum_{k=1}^{K}\cdot |E_k|$ dimensions, while the dual has an effective cost of solving a $Kp$ dimensional optimization, if the selected sizes are of smaller order than $p$. 
\end{remark}

\subsection{Sampler: Langevin random walk}
\label{sampler:langevin}

We describe below a Langevin random walk to sample from the pseudo posterior 
\[\tilde{\pi}_E(\beta^*\lvert S=s) \propto \pi(\beta^*) \cdot \cfrac{f(s\lvert \beta^*)}{\hat{\mathbb{P}}((S,O)\in \mathcal{R}\lvert \beta^*)}\]
post randomized queries based on the approximate normalizer in \eqref{barrier:version}, \eqref{approx:decomposed} and \eqref{dual:barrier}. The method of approximating a target distribution using a Langevin diffusion is studied in \cite{roberts1996exponential}. Another alternative to the simple Langevin sampler implemented in this work, is a Metropolis version with an accept reject step; the afore mentioned reference introduces the ``Metropolis adjusted" version of the algorithm. Depending on the regime of inference, we require the log-MGFs of the generative density and the randomization density for solving the approximating optimization in its dual form or the convex conjugates of the log-MGFs while solving for the primal. A new update ${\beta^*}^{(K)}$ based on a Langevin random walk with target as the pseudo selective posterior $\tilde\pi_E(\beta^*\lvert S)$ is given by
\begin{equation}
\label{new:update}
{\beta^*}^{(K)} = {\beta^*}^{(K-1)} + \eta \grad\log \tilde{\pi}_E\left({\beta^*}^{(K-1)}\lvert S\right) + \sqrt{2\eta}\epsilon^{(K)}
\end{equation}
where $\eta$ is the step-size and $\epsilon^{(K)} \sim \mathcal{N}(0, I)$. This allows us to provide sample-based effect size estimates in the form of credible intervals and point estimates for any function of the parameter of interest $\beta^*$ in the generative model. 

All that the sampler in \eqref{new:update} requires is calculating the gradient of the log-posterior $\tilde{\pi}_E$ as a function of each new draw ${\beta^*}^{(K)}$. 
For a Gaussian generative model on data vector $S$ with mean parametrized as $\mu(\beta^*)$, the below theorem shows that the gradient of the log-pseudo posterior can be computed in terms of the optimizer to the problem in \eqref{barrier:version}. 
\begin{theorem}
\label{gradient:posterior}
The gradient of the log-pseudo \textit{selective posterior} $\log\tilde{\pi}_E(.\lvert S=s)$ at ${\beta^*}^{(K)}$ for a Gaussian generative density for data vector $S$ with mean 
$\mu(\beta^*): \real^k \to \real^d$
 and a variance-covariance matrix $\Sigma_f$ given by
\[f(s\lvert \beta^*) = \frac{1}{(2\pi)^{d/2} |\Sigma_f|^{1/2}} \cdot \exp\left(-(s-\mu(\beta^*))^T\Sigma_f^{-1}(s-\mu(\beta^*))/2\right)\]
with respect to parameter $\beta^*$ is given by
\begin{equation}
\label{grad}
\cfrac{\partial\log\pi(\beta^*)}{\partial \beta^*}\Big\lvert_{{\beta^*}^{(K)}}+ \left(\cfrac{\partial\mu}{\partial \beta^*}\right)^T\Bigg\lvert_{{\beta^*}^{(K)}} \cdot\Sigma_f^{-1}\left\{s-s^*\left(\Sigma_f^{-1}\mu\left({\beta^*}^{(K)}\right)\right) \right\}
\end{equation}
where $s^*(\Sigma_f^{-1}\mu(\beta^*))$ equals
\[\arg\min_{z\in \mathbb{R}^d}\left( z^T\Sigma_f^{-1}\mu(\beta^*) -\frac{1}{2}z^T \Sigma_f^{-1}\ z - \inf_{o\in \mathbb{R}^p}\left\{ \Lambda_g^*(Dz + Po +q) + b_{\mathcal{R}_O}(o)\right\}\right ).\]
\end{theorem}

For the dual optimization in \eqref{dual:barrier} , the optimizer $s^*$ can be derived from the K.K.T. conditions as
\[s^* = \grad \Lambda_f(D^T u^*\lvert \beta^*)\]
where $u^*$ is the dual variable that optimizes
\[\arg\min_{u \in \mathbb{R}^p} \Big\{ \Lambda_f(D^T u\lvert \beta^*)+ \Lambda_g(-u)+ b_{\mathcal{R}_O}^*(P^T u) + u^T q\Big\}\]

This shows that all we need for inference is the solution to the optimization problem cast as \eqref{barrier:version}, \eqref{approx:decomposed} and \eqref{dual:barrier} at each fresh draw.
See Appendix \ref{A:1} for a proof of Theorem \ref{gradient:posterior}. 

\begin{remark}
\emph{\it Estimating equation for MAP:}
It is easy to see that equating \eqref{grad} to $0$ gives rise to an estimating equation for the selective MAP for $\beta^*$. It gives rise to a convex objective for the MAP problem for any log-concave prior $\pi$ on $\beta^*$ and a generative mean $\mu(.)$ that is linear in $\beta^*$. Lemma \ref{selective:MLE} gives the selective MLE under a non-informative prior $\pi\propto 1$ for a Gaussian density with natural parameter as $\beta^*$. A standard gradient descent can be performed on the log-posterior to solve for the MAP in such cases. The pseudo selective MAP in the non-randomized scenario and the simple additive randomized settings is introduced in \cite{selective_bayesian}. 
\end{remark}

\begin{lemma}
\label{selective:MLE}
Under a Gaussian generative density for data vector $S$ considered in \ref{gradient:posterior} with mean parametrized as $\mu(\beta^*)= \beta^*$ and $\Sigma_f= I$, the approximate selective MLE ${\beta^*}^{\text{MLE}}$ based on the pseudo truncated law $\tilde\ell_E(\cdot\lvert \beta^*)$ given by
\begin{equation*}
\begin{aligned}
\log \tilde\ell_E(s\lvert \beta^*) &\propto -s^T s/2 + {\beta^*}^T s-\Gamma(\beta^*) \text{ and }
\end{aligned}
\end{equation*}
\begin{equation*}
\Gamma(\beta^*) = \sup_{z\in \real^d}\Big\{ z^T\beta^*  -\frac{1}{2}z^T z -\inf_{o\in \mathbb{R}^p}\left\{ \Lambda_g^*(Dz + Po +q) + b_{\mathcal{R}_O}(o)\right\}\Big\}
\end{equation*}
satisfies 
\begin{equation}
\grad \Gamma\left({\beta^*}^{\text{MLE}}\right) = s. 
\end{equation}
\end{lemma}
The proof of this is straight-forward from the estimating equation in \eqref{grad}.
In the following section, we show that the approximate normalizer in Theorems \ref{approximate:prob} and \ref{bound:modified} give a valid exponential rate to the selection probability on a large deviation scale under a Gaussian randomization and a Gaussian generative density. Under these conditions, the selective MLE obtained by maximizing the pseudo truncated law in Lemma \ref{selective:MLE} is consistent for $\beta^*$.

%--------------------------------------------------
\section{Limiting approximation on large deviation scale}
%--------------------------------------------------
\label{LDP}
%We show that the approximation to the normalizer in Theorems \ref{approximate:prob} and \ref{bound:modified} give a valid exponential rate to the selection probability on a large deviation scale under a Gaussian randomization density and conditions that allow a large deviation principle for the data vector. 
We fix some notations that apply to this section. In the implementations in Section \ref{experiments}, the columns of the predictor matrix $X$ are normalized by a factor of $1/\sqrt{n}$. We introduce the suppressed scale and denote the normalized $X$ as $X/\sqrt{n}$ in this section. The optimization for such a query is described in details in \ref{lasso:variants} under Section \ref{sampler}. The set-up is similar to the randomized logistic lasso query considered in \cite{randomized_response}, except that we use the Lasso query instead of the logistic Lasso query to illustrate the results in this section. With data
\[y_i, X_i\sim \mathbb{P}_n(\beta_{E,n})\text{  for } i=1,2,...,n\]
where $\mathbb{P}_n\in\{\mathbb{F}_n: \mathbb{E}_{\mathbb{F}_n}[y_i\lvert X_i] = X_{i,E} \beta_{E,n},\mathbb{E}_{ \mathbb{F}_n}[y_i ^2\lvert x_i]-\mathbb{E}^2_{\mathbb{F}_n}[y_i\lvert X_i]=1 \},$
the Lasso query is given by
\begin{equation*}
\label{lasso:randomized:program:scaled}
\argmin_{\beta\in \real^p} \frac{1}{2}\|y - {X\beta}/{\sqrt{n}}\|_2^2 -\omega^T \beta + \lambda\|\beta\|_1 +\frac{\epsilon}{2}\|\beta\|_2^2.
%\vspace{-0.8em}
\end{equation*}
The tuning parameter is set at a theoretical value
$\lambda_n = \mathbb{E}\left[{X^T\psi}/{\sqrt{n}}\right] $
for $\Psi \in \real^p \sim \mathcal{N}(0, I_p)$ and $\epsilon =1/\sqrt{n}.$ 
Denote the scaled versions of the data vector and optimization variables as $S_n$ where
\[S_n =\sqrt{n} \begin{pmatrix}   \left({X_E^T X_E}/{n}\right)^{-1} \cfrac{X_E^T y}{n} \\ \cfrac{X_{-E}^T}{n}\left(y-   X_E\left({X_E^T X_E}/{n}\right)^{-1} \cfrac{X_{E}^T y}{n}\right)\end{pmatrix} = \sqrt{n} \bar{S}_n\]
with $\bar{S}_n$ as the mean of data variables $S_{i,n}(y_i, X_i), \;i=1,2, \cdots, n$ such that
\begin{equation}
\label{mean}
\mathbb{E}_{\mathbb{P}_n}[S_{i,n}(y_i, X_i)\lvert \beta_E] = \mu_n(\beta_{E,n})= \begin{pmatrix}  \beta_{E,n} \\ \mathbb{E}_{\mathbb{P}_n}\left[X_{i, -E}^T (y_i- X_{i, E}\beta_{E,n}) \right] \end{pmatrix}.
\end{equation}
Unlike \cite{randomized_response} which assumes local alternatives of the form $\beta_{E,n}= o(n^{-1/2})$, the selection probability is on the scale of a large deviation probability if $\beta_{E,n} = O(1)$. To simplify notations, we denote $\beta_{E,n}= \beta_E$ hereafter. 

We assume that the randomization instance $\omega$ in \eqref{lasso:randomized:program:scaled} is from a Gaussian density, which is used in all the experiments in Section \ref{experiments}.
The infinite divisibility property of Gaussian densities allows us to write perturbation $\omega = \sqrt{n}\bar{\omega}_n$ where
$\bar{\omega}_n$ is the mean of $n$ i.i.d. Gaussian variables $\omega_{i}, i=1,2,\cdots, n$. The tuning parameter $\lambda_n$ converges to a constant; thus, we can treat it as a constant and denote it as $\lambda$. Also, noting that $X^T X/n$ converges in probability to a constant, we can consider the matrices $D_n, P_n, q_n$ in the inversion map for \eqref{lasso:randomized:program:scaled} as fixed. We use notations $D, P, q$ for the affine inversion map. Finally, let
\begin{equation}
\label{opt:variables}
\sqrt{n}\bar{O}_n = P^{-1}(\sqrt{n}\bar{\omega}_n - \sqrt{n}D\bar{S}_n -q)
\end{equation}
based on the inversion map, with $\bar{O}_n$ interpreted as the mean of
\[O_{i,n} =  P^{-1}(\omega_{i} - D\bar{S}_n-q/\sqrt{n}), \;i=1,2,\cdots, n.\]

Theorem \ref{LDP:thm} gives the limiting rate of decay of the volume of a compact and convex selection region
$\mathcal{R} = \mathcal{R}_S^{'} \times \mathcal{R}_O^{'} \text{ for } \mathcal{R}_S^{'} \subset \real^p \text{ and } \mathcal{R}_O^{'} \subset \real^p$
with respect to the probability density of the augmented vector $(\bar{S}_n, \bar{O}_n)$, whenever the data vector mean satisfies a large deviation principle. Define
\begin{equation}
\label{limit:MGF}
\Lambda_f(\lambda\lvert \beta_E)= \lim_n \frac{1}{n}\Lambda_{\mathbb{P}_n}(n\lambda) = \lim_n \frac{1}{n}\log\mathbb{E}_{\mathbb{P}_n}[\exp(n\lambda^T\bar{S}_n )\lvert \beta_E]
\end{equation}
with $\bar{S}_n$ as the mean of the data vector array $S_{i,n},\;i=1,2,\cdots n$ satisfying \eqref{mean} for $\beta_E\in \real^{|E|}$ and 
\[\mathcal{D} = \{\lambda \in \real^p: \Lambda_f(\lambda\lvert \beta_E)<\infty\}.\]

\begin{theorem}
\label{LDP:thm}
Whenever the limiting log-MGF sequence $\Lambda_f(\lambda\lvert \beta_E)<\infty$ in a neighborhood around $0$ in $\real^p$, $\Lambda_f(\lambda\lvert \beta_E)$ is lower semi-continuous and differentiable in $\mathcal{D}^{0}$ and for any $\lambda \in \partial\mathcal{D}$, $\lim_{\gamma\to \lambda}|\grad \Lambda_f(\nu\lvert \beta_E)|=\infty$,
%Denoting $\Lambda_f$ as the log-MGF of the data vector $S_{1,n}$ in the triangular array  
%\[S_{i,n}\sim \mathbb{F}_n(\beta_E),\;i=1,2,\cdots n\]
 %with mean $\mu(\beta_E)$ as in \eqref{mean} and the corresponding conjugate as $\Lambda_f^*$, 
 the following hold for a compact and convex selection region 
$\mathcal{R} =\mathcal{R}_S^{'} \times \mathcal{R}_O^{'}.$ 
\begin{enumerate}[(1).]
\item Denoting the log-MGF of Gaussian randomization $\omega_{1}$ as $\Lambda_g(.)$ with conjugate $\Lambda_g^*(.)$ and the conjugate corresponding to $\Lambda_f(.\lvert \beta_E)$ in \eqref{limit:MGF} as $\Lambda_f^*(.\lvert \beta_E)$
\begin{equation*}
\begin{aligned}
&\lim\limits_{n} \dfrac{1}{n}\log \mathbb{P}(\bar{S}_n  \in \mathcal{R}_S^{'}, \bar{O}_n \in \mathcal{R}_O^{'}\lvert \beta_E) +\inf_{s\in \mathcal{R}_S^{'}, o \in \mathcal{R}_O^{'}} \{\Lambda_f^*(s\lvert \beta_E) \\
&\;\;\;\;\;\;\;\;\;\;\;\;\;\;\;\;\;\;\;\;\;\;\;\;\;\;\;\;\;\;\;\;\;\;\;\;\;\;\;\;\;\;\;\;\;\;\;\;\;\;\;\;\;\;\;\;\;\;\;\;\;\;\;+ \Lambda_g^*(Ds + Po +q/\sqrt{n}) \}=0.
\end{aligned}
\end{equation*}
\item If the Gaussian randomization density supported on $\real^p$ is independent in all $p$ coordinates with the conjugate of the log-MGF corresponding to active coordinates denoted as $\Lambda_{g_E}^*$ and the selective constraints on the optimization variables are separable as in Theorem \ref{bound:modified}, then 
\begin{equation*}
\begin{aligned}
&\lim_n \frac{1}{n}\log \mathbb{P}(\bar{S}_n  \in \mathcal{R}_S^{'}, \bar{O}_n \in \mathcal{R}_O^{'}\lvert \beta_E) + \inf_{s\in \mathcal{R}_S^{'}, o_E\in \mathcal{R}_E^{'}} \Big\{\Lambda_f^*(s\lvert \beta_E) + \\
&\;\;\;\;\;\;\;\;\;\;\;\;\;\;\;\;\;\;\;\;\;\;\;\;\;\;\;\;\;\;\;\;\;\;\;\;\Lambda_{g_E}^*(D_E s + P_E o_E +q_E/\sqrt{n})-\frac{1}{n}\log \mathcal{B}(o_E;s)\Big\} =0 
\end{aligned}
\end{equation*}
with
\[\Scale[0.98]{\mathcal{B}(o_E;s) =  \prod\limits_{j=1}^{p-|E|}\int\limits_{\mathcal{R}^{'}_{j,-E}}g_{j,-E}(o_{j,-E} +D_{j,-E}s + P_{j,-E} o_{E}+ q_{j,-E}/\sqrt{n})do_{j,-E}.}\]
\end{enumerate}
\end{theorem}
We use in the above theorem the fact that the mean vector $\bar{S}_n$ satisfies a large deviation principle with rate function $\Lambda_f^*$ 
\[\lim_n \frac{1}{n}\log\mathbb{P}(\bar{S}_n \in \mathcal{R}^{'}_S\lvert \beta_E) = -\inf_{s\in \mathcal{R}^{'}_S} \Lambda_f^*(s).\]
Similarly, the conditional probability of $\bar{O}_n$ given $\bar{S}_n$ has a limiting large deviation rate expressed in terms of $\Lambda_g^*(h_n(.))$ composed with the affine inversion map $h_n(.): \real^p \to \real^p$ given by
\[h_n(o) = D\bar{S}_n + Po + q/\sqrt{n}.\] 
That is,
\begin{equation}
\label{LDP:rate:optimization} 
\lim_n \cfrac{1}{n}\log\mathbb{P}(\bar{O}_n \in \mathcal{R}^{'}_O\lvert \bar{S}_n=s) + \inf_{o\in \mathcal{R}_O^{'}} \Lambda_g^*(Ds + Po + q/\sqrt{n})=0.
\end{equation}

This is a consequence of the observation that the limiting rate function is the conjugate of $\lim_n \frac{1}{n} \log\mathbb{E}(\exp(n\lambda^T \bar{O}_n)\lvert \bar{S}_n=s)$; the change of measure map yields the following
\[\lim_n \frac{1}{n} \log\mathbb{E}(\exp(n\lambda^T \bar{O}_n)\lvert \bar{S}_n=s) + \lambda^T P^{-1}(Ds + q/\sqrt{n})-\frac{\lambda^TP^{-1}\Sigma_g{P^{-1}}^{T} \lambda}{2}=0\]
for Gaussian randomization with variance $\Sigma_g$. The proof then follows by an application of the below Lemma \ref{varadhan:lemma}, a modified version of Varadhan's Lemma (see \cite{dembo1998large}). The smooth unconstrained optimizations in \eqref{barrier:version} and \eqref{approx:decomposed} with a continuous barrier penalty function, scaled appropriately also approximate the selection probability accurately as the sample size grows large. Proofs of the above theorem and Lemma \ref{varadhan:lemma} are included in the Appendix \ref{A:3}. 
\begin{lemma}
\label{varadhan:lemma}
For a sequence of functions $H_n(.)$ that uniformly converge to a continuous function $H$ on a compact, convex set $\mathcal{R}\subset \real^d$, the  limit 
\[\lim_n \frac{1}{n}\log\mathbb{E}[\exp(n H_n(\bar{Z}_n))1_{\bar{Z}_n \in \mathcal{R}}] = -\inf_{z\in \mathcal{R}}\{ \Lambda^*(z) - H(z)\}\]
holds for sequence of variables $\bar{Z}_n\in \real^d$ satisfying a large deviation principle with a rate function $\Lambda^*(.)$.
\end{lemma}

Under a Gaussian generative density parametrized by $\beta_E$, a Gaussian randomization with log-MGF $\Lambda_g(.)$ and a compact and convex selection region 
$\mathcal{R} =\mathcal{R}_S^{'} \times \mathcal{R}_O^{'}$ and the same asymptotic set-up as in Theorem \ref{LDP:thm}, it follows as a consequence that the sequence 
\[\Gamma_n(\beta_E) = \sup_{s\in \mathcal{R}_S^{'}} \left\{s\beta_E -s^T s/2-\inf_{o \in \mathcal{R}_O^{'}}\left\{ \Lambda_g^*(Ds + Po +q/\sqrt{n}) + b_{\mathcal{R}_O}(o))\right\} \right\}\]
in  Lemma \ref{selective:MLE} approximates the exact log-partition function 
\[\Gamma_{\text{exact},n} (\beta_E) = \beta_E^T\beta_E/2 + \frac{1}{n}\log  \mathbb{P}(\bar{S}_n  \in \mathcal{R}_S^{'}, \bar{O}_n \in \mathcal{R}_O^{'}\lvert \beta_E) \text{ as } \]
$$\Gamma_n(\beta_E) -\Gamma_{\text{exact},n} (\beta_E) \to 0, \;n \to \infty.$$
Denote the sequence of selective MLE obtained by maximizing the sequence of pseudo truncated likelihoods as $\beta^\text{MLE}_{n, E}$ that satisfies the estimating equation
\[\grad \Gamma\left(\beta^\text{MLE}_{n, E}\right) = \bar{S}_n. \]
Strong convexity of $\Gamma_n(\cdot)$ with a lower bound $B$ on the indices of convexity leads to the identity
\[\|\beta^\text{MLE}_{n, E} -\beta_E\|_2 \leq \dfrac{1}{B}\cdot  \|\bar{S}_n -\grad \Gamma_n(\beta_E)\|_2.\]
Convergence of the approximate log-partition sequence $\Gamma_n(\cdot)$ to the exact one coupled with the identity above prove consistency of the selective MLE $\beta^\text{MLE}_{n, E}$ using similar arguments as Theorem 7.6 in \cite{selective_bayesian}. 

%--------------------------------------------------
\section{Illustrations of truncated Bayesian approach}
%--------------------------------------------------
\label{sampler}

%This section outlines the steps of a Langevin sampler for providing inference based on the \textit{selective posterior} post randomized queries. In principle, any standard sampler that targets the approximate posterior can be employed. We illustrate the truncated Bayesian approach by revisiting some commonly used selective queries- Lasso with a fixed and random design, forward stepwise, logistic Lasso, carved Lasso, a multi-screening query etc. The examples in the section show the wide applicability of our methods to a wide range of convex queries with different
%losses and penalties. 
We illustrate truncated Bayesian approach by revisiting some popular selective queries. These examples are discussed in the context of frequentist inference in \cite{harris2016selective}. In all the below examples, the generative law on the data vector is a Gaussian with mean parametrized as $\beta^*$. In particular, we assume that 
$\mathbb{E}_f[Y\lvert X] = X^*\beta^*.$ We assume that the columns of the design matrix $X$ are scaled by $1/\sqrt{n}$ and denote the scaled predictor matrix as $X$, suppressing the scale $n$. In particular, we assume independent Gaussian entries for $X$. The randomized queries are conducted using instances of randomization from a Gaussian density supported on $\mathbb{R}^p$ with mean $0$ and variance $\tau^2 I_p$.
\begin{remark}
\emph{\it Prior information on parameters:}
We provide inferential results based on the selected model in Section \ref{experiments} where $X^* = X_E$ and $\beta^* =\beta_E \in \real^{|E|}$ under a non-informative prior. Our methods however, do allow the analyst to elicit a prior from an expert or prior experiments post the selective analysis. Our simulation results show that in the absence of an informative prior, the analyst can still capitalize upon the merits of a Bayesian machinery to provide valid inference post selection. 
\end{remark}
The generic recipe for inference using the proposed methods is to compute the inversions maps and selection regions that characterize the output of a query. This is followed by solving the optimization problem for each draw $\beta^{(K)}$ of the sampler. A function of the optimal data vector gives the gradient of the approximate log-posterior at $\beta^{(K)}$ in Theorem \ref{gradient:posterior}; thus, we sample from a tractable version of the selective posterior to carry out Bayesian inference. 
For each of the below examples, we give an explicit approximating optimization based on the inversion map and selective constraints that characterize the randomized query. 

We present below the canonical Lasso query with the $\ell_1$-penalty; we show simulations using both the primal and dual optimizations and a carved version of the Lasso query in Section \ref{experiments}. We describe the optimizations for forward stepwise in \ref{FS} and the thresholding query that is a screening stage of a multi-query in \ref{ms:lasso}, these are examples of popular queries with penalties different from an $\ell_1$ penalty.. Other natural extensions of the Bayesian approach include the grouped selection of variables with a group Lasso penalty as in \cite{loftus2015selective}. We can also apply our methods to other interesting examples such as inference post selection of edges representative of the conditional dependence structure of variables via the graphical Lasso; frequentist selective inference in such a model has been addressed in \cite{penalized_l1}. We do not explore these extensions here.
 
\subsection{A Lasso query}
\label{lasso:variants} 
A randomized version of Lasso with design $X$ based on data vector $S= y \in \mathbb{R}^n$ and randomization instance $\omega \sim \mathcal{N}(0,\tau^2 I_p)$ solves 
\begin{equation*}
\label{lasso:randomized:program}
\argmin\displaystyle_{\beta} \frac{1}{2}\|y - X\beta\|_2^2 -\omega^T \beta + \lambda\|\beta\|_1 +\frac{\epsilon}{2}\|\beta\|_2^2;
%\vspace{-0.8em}
\end{equation*}
to give output $(E,z_E)$, the active set with active signs. 

The selection region imposed by the $\ell_1$-constrained query takes the completely separable form of orthants for active constraints and intervals for inactive constraints; that is
$\mathcal{R}= \mathbb{R}^d \times \mathcal{R}_O$ 
where
\[\mathcal{R}_O=\prod_{j=1}^{|E|} \{o_{j,E}: \text{sign}(o_{j,E}) = z_{j,E}\} \times \prod_{j=1}^{p-|E|}\{o_{j,-E}: |o_{j,-E}|\leq \lambda\}.\]

\noindent{\textbf{\textit{Inversion map}}}: The inversion map encoding selection output $(E, z_E)$ is given by
\begin{equation}
\label{map:lasso}
\begin{aligned}
\omega(y; o) &= -\begin{pmatrix}[1.2]X_E^T \\ X_{-E}^T\end{pmatrix}y + \begin{bmatrix}[1.2]X_E^TX_E +\epsilon I  & 0\\ X_{-E}^TX_E & I\end{bmatrix}o + \begin{pmatrix} \lambda z_E \\ 0\end{pmatrix}= Dy + Po+ q \nonumber \\
&=\Scale[0.95]{ \begin{pmatrix}[1.2]-X_E^T y + (X_E^TX_E +\epsilon I )o_E + \lambda z_E \\  X_{-E}^T y +  X_{-E}^TX_E o_{E} + o_{-E}\end{pmatrix} =  \begin{pmatrix}[1.2] D_E y + P_E o_E + q_E \\ D_{-E}y + P_{-E}o_E + q_{-E} + o_{-E}\end{pmatrix} .}
\end{aligned}
\end{equation}

Based on the above inversion map, the approximating optimization with $n+|E|$ and $p$ optimizing variables in the primal and dual formulation respectively can be computed as below. 
\medskip

\noindent{\textbf{\textit{Primal problem}}}: Under the linear model with mean $X^*\beta^*$ and covariance matrix $\sigma^2 I_n$, the approximation to $\log \hat{\mathbb{P}}((S,O) \in \mathcal{R}\lvert \beta^*)$ leading to pseudo posterior $\tilde{\pi}_E(\beta^*\lvert y)$ is 
\begin{equation*}
\begin{aligned}
&-\inf\limits_{s \in \mathbb{R}^{n},\;o_E \in \mathbb{R}^{|E|}}\Big\{  \frac{1}{2\sigma^2} \|s-X^*\beta^*\|_2^2 +\frac{1}{2\tau^2}\|D_E s +P_E o_E+ q_E\|_2^2\\
&\;\;\;\;\;\;\;\;\;\;\;\;\;\;\;\;\;\;\;\;\;\;\;\;\;\;\;\;\;\;-\log\mathcal{B}(o_E;s) + b_{\mathcal{R}_E}(o_E)\Big\}
\end{aligned}
\end{equation*} 
where the volume of the inactive selection region conditional on $S=s, O_E= o_E$ under the isotropic Gaussian randomization is computed as 
\[\log\mathcal{B}(o_E;s) = \sum\limits_{j=1}^{p-|E|} \log \left\{\Phi\left(\frac{\lambda+ \alpha(o_E; s)_j}{\tau}\right) - \Phi\left(\frac{-\lambda+ \alpha(o_E; s)_j}{\tau}\right)\right\}\]
with $\alpha(o_E; s) = D_{-E}s + P_{-E} o_{E}+ q_{-E}$
and $ D_{-E}, P_{-E}, q_{-E}$ as given in \eqref{map:lasso}.
\medskip

\noindent{\textbf{\textit{Dual problem}}}: The corresponding dual in terms of the logarithm of the Gaussian MGFs (both for data and randomization) and conjugate of the barrier function is given by
\[\inf_{u \in \mathbb{R}^p} \frac{1}{2}\sigma^2 \|D^T u\|_2^2 + \frac{1}{2}\tau^2 \|u\|_2^2 + b^*_{\mathcal{R}_O}(P^T u) + u^T (D X^*\beta^* +q),\]
where $D, P, q$ are obtained from the map in \eqref{map:lasso}.
Calculation of the conjugate of the barrier function follows from Appendix \ref{dual:details}.
A carved query solves the randomized version of lasso described in \eqref{lasso:randomized:program} on a random split of the data $S(X^{(1)}, y^{(1)})$ leading to the output $(E, z_E)$. Such a query takes the form 
\[ \argmin_{\beta} \frac{1}{2r}\|y^{(1)} - X^{(1)}\beta\|_2^2 + \lambda\|\beta\|_1 +\frac{\epsilon}{2}\|\beta\|_2^2\]
\medskip
where $r$ is the fraction of the data used in the above selective query. 

\noindent{\textbf{\textit{Inversion map}}}:
The randomization inherited from the random split on the data, as described in \cite{markovic2016bootstrap} leads to the below inversion map
\[\omega(s,o) = \partial \ell(s; (\hat\beta_E,0))-\frac{1}{r}\partial \ell(s^{(1)}; (\hat\beta_E,0)).\]
The randomization described as above
is asymptotically Gaussian with mean $0$ and covariance $\Sigma_g$ and is asymptotically independent of the data vector $S$ described for the random $X$ lasso query.
%\[S= \begin{pmatrix}[1.2] (X_E^T X_E)^{-1} X_E^T y \\ X_{-E}^T(I -\mathcal{P}_E) y\end{pmatrix}.\] 
\medskip

\noindent{\textbf{\textit{Primal problem}}}: While we can no longer use the reduced optimization in \eqref{approx:decomposed} (as the randomization inherited from the split is not independent in all p-coordinates), we can use the more general approximation to the normalizer in \eqref{barrier:version}. The joint on the data and randomization is an asymptotic Gaussian, with the data mean parametrized as $\mu(\beta^*)$. The approximating optimization that we solve to sample from the pseudo \textit{selective posterior} is given by
\begin{equation*}
\begin{aligned}
&\Scale[0.96]{-\inf\limits_{s \in \mathbb{R}^{p},\;o \in \mathbb{R}^{p}}\Big\{  \frac{1}{2}(s-\mu(\beta^*))^T\Sigma_f^{-1}(s-\mu(\beta^*)) +\frac{1}{2}(D s +Po+ q)^T \Sigma_g^{-1}
(D s +P o+ q)}\\
&\;\;\;\;\;\;\;\;\;\;\;\;\;\;\;\;\;\;\;\;\;\;\;\;\;\;\;\;\;\;\;+ b_{\mathcal{R}_O}(o)\Big\}
\end{aligned}
\end{equation*} 
with $\Sigma_g$ and $\Sigma_f$ estimated by bootstrap.
%\end{example}

\subsection{A forward stepwise query}
\label{FS}
We discuss the approximating optimization that we solve to give truncated Bayesian inference after 2 steps of forward stepwise selection (FS) next. This can be easily generalized to $K$ steps. In Section \ref{experiments}, we give adjusted estimates in a Bayesian model after $1$ step of FS. This can also be viewed as a sequential query on the data. 
%selecting the most correlated variable adjusted for previously selected ones. % such a query on the data solves \eqref{canonical:randomized:program} with loss given by 
%$\ell(X,y; \beta) = -\eta^T X^T y$ and penalty as the characteristic function of $\ell_{1}$ norm of radius $1$ 
%\[ I^{(1)}_{\ell_{1}}(\beta)=\begin{cases} 
%      0 & \|\beta\|_{1}\leq 1 \\
%      \infty & \text{otherwise.}
%   \end{cases}
%\]
\medskip

\noindent{\textbf{\textit{Inversion maps}}}: 
Denoting $E_1=\{j_1\}$ and $E_2=\{j_1,j_2\}$ and the predictor for second stage as $\tilde{X} = \mathcal{P}_{j_1}^{\perp}X_{-j_1}$ (adjusted for selection of $j_1$ in the first step), 
%with $j_2$ the second most correlated index obtained upon solving 
%\begin{equation*}
%\min_{\beta\in \mathbb{R}^{p-1}}-\beta^T X_{-E_1}^T\mathcal{P}_{E_1}^{\perp} y -\beta^T \omega + I^{(1)}_{\ell_{1}}(\beta)
%\end{equation*}
the characterizing inversion maps for the two-stage sequential selection procedure are as below.
\begin{equation*}
\omega_1=\begin{pmatrix}[1.2] -X_{j_1}^T s + o_{j_1}  \\ -X_{-j_1}^T s +  o_{-j_1}\end{pmatrix};\omega_2=\begin{pmatrix}[1.2] -\tilde{X}_{j_2}^T s + o_{j_2}  \\ -\tilde{X}_{-j_2}^T s +  o_{-\{j_1,j_2\}}\end{pmatrix}.
\end{equation*}
giving selection regions
\[{\mathcal{R}_{O_1} = \{(o_{j_1}, o_{-j_1})\in \mathbb{R}^p: z_{j_1} o_{j_1} \geq 0, \|o_{-j_1}\|_{\infty} \leq |o_{j_1}|\},}\]
\[{\mathcal{R}_{O_2} = \{(o_{j_2}, o_{-\{j_1,j_2\}})\in \mathbb{R}^{p-1}: z_{j_2} o_{j_2} \geq 0, \|o_{-\{j_1,j_2\}}\|_{\infty} \leq |o_{j_2}|\}}\]
where $z_{j_1}$ and $z_{j_2}$ represent the signs of the active variables entering the model in the first and second steps respectively.
\begin{remark}
\emph{\it Selection regions in FS:} As we can see from above that the selection regions in this example take the form of a cone rather than the usual orthant and cube yielded by the $\ell_1$ penalty in the variants of Lasso. We can still write the selection region as
$\mathcal{R}_{O_1}= \mathcal{R}_{E_1} \times \prod_{j\neq j_1}\mathcal{R}_{j,-E_1}$
where \[\mathcal{R}_{j,-E_1} =\{ o_{j,-j_1} : |o_{j,-j_1}| \leq |o_{j_1}|\} \text{ for } j=\{1,2,...,p\} \setminus E_1.\] That is the inactive selective constraints are all separable in the $p-1$ coordinates, although they are determined by the active optimization variable unlike the example of Lasso. The probability of the inactive optimization variables being constrained to be smaller in magnitude than $|o_{j_1}|$ can be computed exactly as $\mathcal{B}( o_{j_1}; s)$ conditional on data $s$ and active variable $ o_{j_1}$. A similar computation goes through for more than $1$ step of FS.
\end{remark}

\noindent{\textbf{\textit{Primal problem}}:}  Denoting the separable inactive selection regions $\mathcal{R}_{j,-E_1} =\{ o_{j,-j_1} : |o_{j,-j_1}| \leq |o_{j_1}|\}$ for $j=\{1,2,...,p\} \setminus E_1$ and 
$\mathcal{R}_{j,-E_2} =\{ o_{j,-\{j_1,j_2\}} : |o_{j,-\{j_1,j_2\}}| \leq |o_{j_2}|\}$ for $j=\{1,2,...,p\} \setminus E_2$, solve $2p-3$ univariate Gaussian probabilities as $\mathcal{B}( o_{j_1}; s)$ and $\mathcal{B}( o_{j_2}; s)$ where
\begin{equation*}
\label{cube:problem:ms}
\begin{aligned}
& \log\mathcal{B}( o_{j_1}; s) =\sum\limits_{j\neq j_1}\Bigg\{\log\Phi\left(({o_{j_1}-X_{j,-j_1}^T s})/{\tau}\right)-\log\Phi\left(({-o_{j_1}-X_{j,-j_1}^T s})/{\tau}\right)\Bigg\}
 \end{aligned}
\end{equation*}
\begin{equation*}
\label{cube:problem:lasso}
\begin{aligned}
& \log\mathcal{B}( o_{j_2}; s) =\sum\limits_{j\neq j_1, j_2}\Bigg\{\log\Phi\left((o_{j_2}-\tilde{X}_{j,-j_2}^T s)/{\tau}\right)- \log\Phi\left((-o_{j_2}-\tilde{X}_{j,-j_2}^T s)/{\tau}\right)
\Bigg\}.
\end{aligned}
\end{equation*}
Solve an optimization over $(s,o_{j_1}, o_{j_2})$ where $s\in \mathbb{R}^n, o_{j_1} \in \mathbb{R}, o_{j_2} \in \mathbb{R}$ with sign barriers $b_{\mathcal{R}_{j_1}}(.)$ and $b_{\mathcal{R}_{j_2}}(.)$ on $o_{j_1}$ and $o_{j_2}$, that is:
\begin{equation*}
\label{primal:problem:ms:lasso}
\begin{aligned}
&\inf_{s,o_{j_1},o_{j_2}}\Big\{\dfrac{\|s- X^*\beta^*\|^2}{2\sigma^2} + \dfrac{\|-X_{j_1}^T s + o_{j_1}\|^2}{2\tau^2} -\log\mathcal{B}( o_{j_1}; s)  + b_{\mathcal{R}_{E_1}}(o_{j_1})\\
&\;\;\;\;\;\;\;\;\;\;\;\;\;\;\;\; +\dfrac{\|\tilde{X}_{j_2}^T s + o_{j_2}\|^2}{2\tau^2} - \log\mathcal{B}( o_{j_2}; s)   + b_{\mathcal{R}_{E_2}}(o_{j_2})\Big\}.
 \end{aligned}
\end{equation*}
\noindent{\textbf{\textit{Dual problem}}:}
Solve a dual optimization over $u_1 \in  \mathbb{R}^{p}$ and $u_2 \in \mathbb{R}^{p-1}$, as stated below:
\begin{equation*}
\label{dual:problem:fs:2step}
\begin{aligned}
&{\inf\limits_{u_{1},u_{2}}\Big\{ \frac{1}{2}\sigma^2{\|X u_1 + \tilde{X} u_2\|^2}-(X u_1 + \tilde{X} u_2)^T X^*\beta^*}\\
&{\;\;\;\;\;\;\;\;\;\;\;\;+ \frac{1}{2}\tau^2 u_1^2+ \frac{1}{2}\tau^2 u_2^2+ b_{\mathcal{R}_{O_1}}^*(u_1)+ b_{\mathcal{R}_{O_2}}^*( u_2)\Big\}}.
\end{aligned}
\end{equation*}

\subsection{A 2-stage query: thresholding followed by Lasso} 
\label{ms:lasso}
We present an example of a two-stage screening method in the linear regression setting with a fixed design matrix $X$ with normalized columns; we derive the approximating optimization problem to provide inference in a Bayesian model with prior $\pi$ on $\beta^*$ and $Y\lvert \beta^*\sim \mathcal{N}(X^*\beta^*, \sigma^2 I)$. The selective analysis comprises of two stages of screening based on realizations $\omega_1, \omega_2$ from independent Gaussian distributions, each with all i.i.d. mean $0$ components and variance $\tau^2 I$. The first query is a randomized marginal screening across the $Z$-statistics at a nominal threshold vector $\alpha$, that solves
\begin{equation*}
\label{ms:randomized:program}
\min_{\beta} \frac{1}{2}\|\beta - X^T y/\hat{\sigma}\|_2^2 -\omega^T \beta + I^{\alpha}_{\ell_{\infty}}(\beta); \text{ with }  
\end{equation*}
\[ I^{\alpha}_{\ell_{\infty}}(\beta)=\begin{cases} 
     0 & \|\beta\|_{\infty}\leq \alpha \\
      \infty & \text{otherwise.}
   \end{cases}\]
This results in output $(E_1,z_{E_1})$, the active set of marginally most correlated predictors with active signs from Stage-I screening. \\
%followed by a second stage selection that solves a randomized version of Lasso. Fixing notations, we denote by $E_1$ as the active set in the first stage screening and $E_2 \subset E_1$ as the active set obtained in the second round of selection. Denote active and inactive optimization variables as $o_{E_1}, o_{-E_1}$ for stage-I and $o_{E_2}, o_{-E_2}$ for stage-II. \\
Denoting $\tilde{X}= X_{E_1}\in \mathbb{R}^{n} \times  \mathbb{R}^{E_1}$, the predictor matrix with selected predictors from the first round of screening, the second query is a randomized lasso query that solves \eqref{lasso:randomized:program} with design matrix $\tilde{X}$
to yield active set $E_2$ with signs $z_{E_2}$. 

The first step describes the inversion maps and selective constraints encoding the two selective queries where data vector $S=Y$. 
\medskip

\noindent{\textbf{\textit{ Inversion maps}}}:
%Screen p-values based on two-tailed Z (or T) test statistics given by $Z_j=X_j^T y /\hat\sigma \|X_j\|_2$ at an appropriately chosen threshold $\alpha$; based on solving \eqref{canonical:randomized:program} with loss given by 
%$\ell(Z; \beta) = \frac{1}{2}\|\beta - Z(X;y)\|_2^2$ and penalty as the characteristic function of $\ell_{\infty}$ norm at threshold vector $\alpha$
%\[ I^{\alpha}_{\ell_{\infty}}(\beta)=\begin{cases} 
%      0 & \|\beta\|_{\infty}\leq \alpha \\
%      \infty & \text{otherwise.}
%   \end{cases}
%\]
\begin{equation*}
\label{rand:map:ms}
\text{Map I :-- } \omega_1 = \begin{pmatrix} \alpha z_{E_1}-  X_E^T s/\hat\sigma+ o_{E_1}\\ -X_{-E}^T s/\hat\sigma + o_{-E_1}\end{pmatrix} 
 \end{equation*}
 \begin{equation*}
\label{rand:map:ms}
{\text{Map II :-- } \omega_2=\begin{pmatrix}[1.2]-\tilde{X}_{E_2}^T s + (\tilde{X}_{E_2}^T \tilde{X}_{E_2} + \epsilon I)o_{E_2} + \lambda z_{E_2}  \\ -\tilde{X}_{-E_2}^T s + \tilde{X}_{-E_2}^T \tilde{X}_{E_2} o_{E_2} + o_{-E_2}
\end{pmatrix}}
 \end{equation*}
inducing respective selection regions \[\mathcal{R}_{O_i}= \mathcal{R}_{E_i} \times \Pi_j\mathcal{R}_{j,-E_i}; \;i=1,2\] as 
\begin{equation*}
{\mathcal{R}_{O_1} = \prod_{j=1}^{|E_1|}\{ \text{sign}(o_{j, E_1})= z_{j, E_1}\}\times \prod_{j=1}^{p-|E_1|} \{| o_{j,-E_1}|\leq \alpha\}}
\end{equation*} 
\begin{equation*}
\mathcal{R}_{O_2} =  \prod_{j=1}^{|E_2|}\{\text{sign}(o_{j,E_2})= z_{j, E_2}\} \times \prod_{j=1}^{|E_1|-|E_2|} \{| o_{j,-E_2}|\leq \lambda\}.
\end{equation*}  

Using the facts that the convex conjugate of a Gaussian log-MGF with mean $\mu$ and variance $\gamma^2 I_k$ at vector $x$ is 
$\|x-\mu\|_2^2/2\gamma^2$ and the log-MGF is 
$\mu^T x+  \gamma^2\|x\|^2/2$
we derive the primal and dual optimization problems to sample from the approximate posterior. 
\medskip

\noindent{\textbf{\textit{Primal problem}}}: The primal marginalizes over the inactive sub-gradients followed by the optimization over active variables and data in $n +|E_1|+ |E_2|$ dimensions. Computing the exact log-Gaussian probabilities over intervals $[-\alpha, \alpha]$ and $[-\lambda, \lambda] $ as 
\begin{equation*}
\begin{aligned}
\log \mathcal{B}(o_{E_1}; s) &= \sum_{j=1}^{p-|E_1|} \log \Big\{\Phi\left({\{\alpha-X_{j,-E_1}^T s/\hat{\sigma} + o_{j, E_1}\}}/{\tau}\right) \\
&\;\;\;\;\;\;\;\;\;\;- \Phi\left({\{-\alpha -X_{j,-E_1}^T s/\hat{\sigma} + o_{j, E_1}\}}/{\tau}\right)\Big\}
\end{aligned}
\end{equation*}
\begin{equation*}
\begin{aligned}
\log \mathcal{B}(o_{E_2}; s) &= \sum_{j=1}^{|E_1|-|E_2|} \log \Big\{\Phi\left({\{\lambda + \tilde{X}_{j,-E_2}^T \tilde{X}_{E_2} o_{E_2} -\tilde{X}_{j,-E_2}^T s\}}/{\tau}\right) \\
&\;\;\;\;\;\;\;\;\;\;\;\;\;\;\;\;\;- \Phi\left({\{-\lambda+  \tilde{X}_{j,-E_2}^T \tilde{X}_{E_2} o_{E_2} -\tilde{X}_{j,-E_2}^T s\}}/{\tau}\right)\Big\}
\end{aligned}
\end{equation*}
we have the optimization in the primal form as
\begin{equation*}
\label{primal:problem:ms:lasso}
\begin{aligned}
&-\inf\limits_{s,o_{E_1},o_{E_2}}\Bigg\{\dfrac{\|s- X^*\beta^*\|^2}{2\sigma^2} + \dfrac{\|\alpha z_{E_1} -X_{E_1}^T s/\hat{\sigma} + o_{E_1}\|^2}{2\tau^2}-\log \mathcal{B}( o_{E_1}; s)\\
&\;\;+ b_{\mathcal{R}_{E_1}}(o_{E_1}) + \dfrac{\|\tilde{X}_{E_2}^T \tilde{X}_{E_2} o_{E_2} -\tilde{X}_{E_2}^T s +\lambda z_{E_2}\|^2}{2\tau^2}-\log \mathcal{B}( o_{E_2}; s)+ b_{\mathcal{R}_{E_2}}(o_{E_2})\Big\}.
%&\Scale[0.90]{\;\;\;\;\;\;\;\;\;\;\;\;\;\;+ \mathcal{B}( o_{E_2}; s)+ b_{\mathcal{R}_{E_2}}(o_{E_2})\Big\}}
 \end{aligned}
\end{equation*}

\noindent{\textbf{\textit{Dual problem}}}:
With $P_1$ and $P_2$ identified respectively as \[I_p \text{ and } \begin{bmatrix}[1.2] \tilde{X}_{E_2}^T \tilde{X}_{E_2} + \epsilon I & 0 \\ \tilde{X}_{-E_2}^T \tilde{X}_{E_2} & I\end{bmatrix}\]
from the randomization maps, solve an optimization over $u_1 \in  \mathbb{R}^{p}$ and $u_2 \in \mathbb{R}^{|E_1|}$ as below:
\begin{equation*}
\label{dual:problem:ms:lasso}
\begin{aligned}
& {\inf\limits_{u_{1},u_{2}}\Big\{ \frac{1}{2}\sigma^2{\| X^Tu_1/\hat{\sigma} + \tilde{X}^T u_2\|^2}-(X^Tu_1/\hat{\sigma} + \tilde{X}^T u_2)^T X^*\beta^*}+u_1^T \begin{pmatrix} \lambda  z_{E_2} \\ 0\end{pmatrix} \\
&\;\;\;\;\;\;\;\;\;\;\; + u_2^T \begin{pmatrix} \alpha z_{E_1} \\ 0\end{pmatrix} + \frac{1}{2}\tau^2 u_1^2+ \frac{1}{2}\tau^2 u_2^2+ b^*_{\mathcal{R}_{O_1}}(P_{1}^T u_1)+ b^*_{\mathcal{R}_{O_2}}(P_{2}^T u_2)\Big\}.
\end{aligned}
\end{equation*}

%--------------------------------------------------
\section{Experiments}
%--------------------------------------------------
\label{experiments}

\subsection{Simulated models}
\label{simulated:models}
We conduct different experiments to show the coverage and risk properties of estimates, obtained using our methods in comparison to those based on untruncated approach. We vary generative models across our experiments: this highlights that our methods show good performance even under misspecified models. We gives estimates under commonly used selective queries with different losses and penalties.

In the first experiment, we use Model I in Section \ref{motivation} to generate our data. The ground truth is the null model $Y\sim \mathcal{N}(0,I)$. For a fixed design $X$, we draw $Y\in \real^n$ for every repetition using the same $X$. For a random design $X$, we randomly draw $X\in \mathbb{R}^{n \times p}$ with Gaussian entries and draw $Y$ conditional on $X$ and the underlying parameter $\beta$ in each repetition of the experiment. The columns of design $X$ are scaled by $1/\sqrt{n}$ in all cases. The second experiment uses Model II in Section \ref{motivation} as a generative mechanism. This is a Bayesian model with ground truth $\Theta_E(\beta) = (X_E^T X_E)^{-1} X_E^T X \beta$, determined by $E$ in each trial.

In table \ref{table:expt1:I}, we compare the empirical coverage of the credible intervals, the risk of the posterior mean and the length of intervals between the approximating method that aims at the pseudo selective posterior and the usual Bayesian posterior inference. For the Bayesian model, we report the empirical (Bayesian) FCR, Bayes risk of the posterior mean and the length of intervals in table \ref{table:expt1:II}. The queries as described in Examples under \ref{lasso:variants} are conducted under centered Gaussian randomization with variance $\tau^2 I_p$; with the exception of the carved query which inherits randomization from a randomly chosen split of the data. For inference, we use the selected model $\mathcal{N}(X_E\beta_E, I)$ and a non-informative prior on $\beta_E$, where $E$ is the active set from the selective query. In both experiments, we use a misspecified likelihood and prior. Note that despite the fact that the model for inference is a mis-specified one under the true generative models, our methods display superiority in terms of coverage and risk properties in comparison to the unadjusted estimates. The first column states the query- the Lasso with a fixed and random design, a carved Lasso with a random design and 1 step of forward stepwise (FS) and the last column gives the regression dimensions $n$ and $p$. 
\begin{center}
\small
\vspace{-1.0em}
\captionof{table}{Expt 1- Model I}
\label{table:expt1:I}
\bgroup
\def\arraystretch{1.}
\scalebox{0.8}{\begin{tabular}{ |c | c c |c c |c c| c|}
\hline
  & \multicolumn{2}{c}{Coverage} & \multicolumn{2}{c}{Risk} & \multicolumn{2}{c|}{Lengths} & \\
\hline
\bf{Query}  & adjusted & unadjusted  & adjusted  & unadjusted  &  adjusted  & unadjusted & $n,p$\\ [0.5ex] 
 \hline\hline
 \thead{Lasso \\
 (Fixed $X$) Primal} & $86.20\%$  & $22.72\%$ & $1.85$ & $5.43$ & $4.55$ & $3.32$ &\thead{$200$\\$1000$} \\ [0.5ex] 
 \hline
  \thead{ Lasso \\ (Fixed $X$) Dual} & $89.70\%$  & $51.38\%$ & $1.81$ & $3.38$ & $4.41$ & $3.31$ & \thead{$1000$\\ $200$} \\ [0.5ex] 
 \hline
  \thead{Lasso\\  (Random $X$)} & $85.42\%$  & $43.44\%$ & $1.87$ & $3.74$ & $4.41$ & $3.31$  & \thead{$1000$\\ $200$}\\ [0.5ex] 
% \hline
%   \thead{Logistic  Lasso \\ (Random $X$) \\ $n=500, p=100$} & $86.26\%$  & $78.48\%$ & $5.99$ & $6.75$ & $7.08$ & $6.71$  \\ [0.5ex] 
 \hline
\thead{Carved  Lasso \\ (Random $X$) } & $87.30\%$  & $23.16\%$ & $4.54$ & $4.68$ & $6.05$ & $3.32$ & \thead{$1000$\\ $100$} \\ [0.5ex] 
 \hline  
\bf{FS} & $85.15\%$  & $14.85\%$ & $3.46$ & $7.11$ & $4.62$ & $3.30$ & \thead{$200$\\$1000$} \\ [0.5ex] 
 \hline
\end{tabular}}
\egroup
\end{center}
The generative mechanism in the third experiment is a frequentist model that deviates from the all noise model considered in Experiment 1. 
It gives an assessment of estimates based on the output from a Lasso query with a fixed $X$ design using the primal and dual problems by varying the sparsity levels in the true generative mechanism. Based on a fixed predictor matrix, we simulate $Y\in \real^n$ in each draw as below for a sparse vector $\beta_{\mathcal{S}}$ with true support $\mathcal{S}\subset \{1,2,\cdots, p\}$
\[Y\lvert X, \beta_{\mathcal{S}} = X_{\mathcal{S}}\beta_{\mathcal{S}}  + \epsilon, \;  \epsilon \sim \mathcal{N}(0, I_n).\]
We use the primal and dual formulation of the optimization in Example \ref{lasso:variants} for providing estimates in a high dimensional sparse problem $n=500,\; p=3000$ and in the low dimensional regime $n=3000,\; p=500$ respectively. We vary the sparsity levels as $|\mathcal{S}|=0, 5, 10, 20$ signals, each with magnitude $7$. Tables \ref{table:expt2:I} and \ref{table:expt2:II} show that the adjusted estimates have superior risk and coverage properties as compared to the unadjusted estimates, both based on a selected model appended to a diffuse prior.
\begin{center}
\small
\vspace{-1.0em}
\captionof{table}{Expt 2- Model II}
\label{table:expt1:II}
\bgroup
\def\arraystretch{1.}
\scalebox{0.8}{\begin{tabular}{ |c | c c |c c |c c| c|}
\hline
  & \multicolumn{2}{c}{Bayesian CR} & \multicolumn{2}{c}{Bayes risk} & \multicolumn{2}{c|}{Lengths} &   \\
\hline
\bf{Query}  & adjusted & unadjusted   & adjusted  & unadjusted   &  adjusted  & unadjusted   & $n,p$\\ [0.5ex] 
 \hline\hline
 \thead{Lasso \\(Fixed $X$) Primal} & $90.99\%$  & $33.86\%$ & $1.49$ & $4.28$ & $4.49$ & $3.34$ & \thead{$200$\\ $1000$}  \\ [0.5ex] 
 \hline
  \thead{ Lasso \\ (Fixed $X$) Dual} & $87.12\%$  & $61.99\%$ & $1.71$ & $2.71$ & $4.22$ & $3.31$  & \thead{$1000$\\ $200$}\\ [0.5ex] 
 \hline
  \thead{Lasso\\ (Random $X$) } & $88.26\%$  & $55.04\%$ & $1.77$ & $ 3.01$ & $4.23$ & $3.31$ & \thead{$1000$ \\$200$}\\ [0.5ex] 
 \hline
 \thead{Carved  Lasso \\ (Random $X$) } & $82.86\%$  & $38.30\%$ & $5.98$ & $6.75$ & $5.92$ & $3.31$ & \thead{$1000$\\ $100$} \\ [0.5ex] 
 \hline  
 \thead{FS} & $72\%$  & $51.40\%$ & $3.19$ & $3.99$ & $4.07$ & $3.29$ & \thead{$200$\\ $1000$}  \\ [0.5ex] 
 \hline
\end{tabular}}
\egroup
\end{center}

\begin{center}
\small
\vspace{-1.0em}
\captionof{table}{Expt 3- Deviation from noise model}
\label{table:expt2:I}
\bgroup
\def\arraystretch{1.2}
\scalebox{0.88}{\begin{tabular}{ |c | c c |c c |c c|}
\hline
$n=500, p=3000$  & \multicolumn{2}{c}{Coverage} & \multicolumn{2}{c}{Risk} & \multicolumn{2}{c|}{Lengths}  \\
\hline
\bf{Sparsity}  & adjusted & unadjusted   & adjusted  & unadjusted   &  adjusted  & unadjusted \\ [0.5ex] 
 \hline\hline
0 & $85.25\%$  & $23.22\%$ & $2.09$ & $5.56$ & $4.52$ & $3.31$  \\ [0.5ex] 
 \hline
 5 & $87.27\%$  & $54.99\%$ & $1.73$ & $3.35$ & $4.06$ & $3.35$  \\ [0.5ex] 
 \hline
 10 & $86.05\%$  & $63.98\%$ & $1.80$ & $2.76$ & $4.05$ & $3.36$  \\ [0.5ex] 
 \hline
 20 & $83.51\%$  & $69.49\%$ & $1.86$ & $3.73$ & $4.10$ & $3.43$  \\ [0.5ex] 
 \hline
\end{tabular}}
\egroup
\end{center}

\begin{center}
\small
%\vspace{-1.0em}
\captionof{table}{Expt 3-  Deviation from noise model}
\label{table:expt2:II}
\bgroup
\def\arraystretch{1.2}
\scalebox{0.88}{\begin{tabular}{ |c | c c |c c |c c|}
\hline
$n=3000, p=500$  & \multicolumn{2}{c}{Coverage} & \multicolumn{2}{c}{Risk} & \multicolumn{2}{c|}{Lengths}  \\
\hline
\bf{Sparsity}  & adjusted & unadjusted  & adjusted  & unadjusted  &  adjusted  & unadjusted \\ [0.5ex] 
 \hline\hline
0 & $85.98\%$  & $39.94\%$ & $1.96$ & $3.93$ & $4.36$ & $3.30$  \\ [0.5ex] 
 \hline
 5 & $88.23\%$  & $64.68\%$ & $1.44$ & $2.56$ & $3.76$ & $3.30$  \\ [0.5ex] 
 \hline
 10 & $85.06\%$  & $75.4\%$ & $1.48$ & $1.98$ & $3.68$ & $3.30$  \\ [0.5ex] 
 \hline
 20 & $87.1\%$  & $79.4\%$ & $1.45$ & $1.60$ & $3.71$ & $3.31$  \\ [0.5ex] 
 \hline
\end{tabular}}
\egroup
\end{center}
The final experiment gives the performance of the estimates post the 2-stage screening query with a fixed $X$ design, described in Section \ref{ms:lasso}. We again use both the frequentist Model I and the Bayesian Model II to generate our data. We choose to provide inference using an adaptive target under a model, both of which are determined by the final screened model $E_2$ that combines the output from the two screenings. The coverage and risk comparisons for the above screening procedure are given in table \ref{table:expt3}. The first column gives the Model generating the data and the last column gives the dimensions of the simulation. The only case where the Bayes risk of the adjusted estimate is slightly more than that of the unadjusted posterior mean is for Bayesian model when $n=200, \;p=1000$. 
\begin{center}
\small
\vspace{-1.0em}
\captionof{table}{Expt 4: A 2-stage screening}
\label{table:expt3}
\bgroup
\def\arraystretch{1.}
\scalebox{0.86}{\begin{tabular}{  |c| c c |c c |c c|c|}
\hline
  & \multicolumn{2}{c}{Coverage/ Bayesian CR} & \multicolumn{2}{c}{Risk/ Bayes risk} & \multicolumn{2}{c|}{Lengths} & \\
  \hline
 \thead{Model} & adjusted & unadjusted  & adjusted  & unadjusted  &  adjusted  & unadjusted & $n, p$\\ [0.5ex] 
 \hline\hline
I & $86.52\%$  & $39.20\%$ & $3.22$ & $4.08$ & $5.74$ & $3.40$  & $200, 1000$\\ [0.5ex] 
\hline
I & $89.26\%$  & $29.64\%$ & $2.06$ & $4.27$ & $4.89$ & $3.30$  & $1000, 200$\\ [0.5ex] 
\hline
II & $85.00\%$  & $53.88\%$ & $3.71$ & $3.27$ & $5.66$ & $3.43$  & $200, 1000$\\ [0.5ex] 
\hline
II & $91.86\%$  & $32.39\%$ & $2.01$ & $4.10$ & $4.87$ & $3.31$  & $1000, 200$\\ [0.5ex] 
\hline
\end{tabular}}
\egroup
\end{center}

\subsection{Data analysis: inference on causal variants}
\label{egene:data}
To illustrate the inferential gains with the truncated Bayesian method, we provide adjusted effect size estimates for SNPs (Single nucleotide polymorphisms) that have been data-mined as the strongest associations with gene expression. An analyst will be confronted in defending the strength of these associations if she does not overcome the selective bias encountered in estimation post data-snooping. With gene expression data as the outcome, we give adjusted effect size estimates of SNPs that have been selected as the set of probable causal genetic variants. We highlight the differences between the adjusted Bayesian approach and the unadjusted counterpart (that is inadaptive to selection); we also depict the higher statistical power associated with the adjusted Bayesian estimates post a randomized selection as opposed to the estimates based on \cite{exact_lasso} post a non-randomized selection. 

The data analyzed in this work involves gene expression data $Y \in \real^{97}$ for a gene collected from the human tissue -Liver for a sample of 97 densely genotyped individuals. More details on this data-set are included in the Appendix \ref{egene:data:appendix}. The exceedingly small sample size in this analysis does not allow the analyst to reserve a hold-out data set for inference. The goal here, is to quantify the effect sizes of variants that have been selected from a set of $5233$ of potential predictors, namely $X\in \real^{97\times 5233}$ as predictors that best explain the variance in expression levels of the gene under study. More specifically, the columns of $X$ represent local genetic variants measured as SNPs that lie within 1MB of the transcription start site of the gene. This data has been investigated as a part of a genome-wide association study conducted in \cite{gtex2015genotype, ongen2015fast} with focus on identifying the significant associations between gene expression and genetic variants across different human tissues. The afore-mentioned works aimed at recognizing genes with at least one causal variant, called eGenes. A more recent work \cite{aguet2016local} performs a secondary analysis on the eGenes to further identify variants that regulate the expression for these genes. This involves a search over the local variants around the genes. In this work, we employ one such selection procedure, the commonly used Lasso to pick promising predictors and apply our method to give estimates for effect sizes of these selected SNPs based on the truncated posterior. 

Below, we outline the analysis that leads to the selection of SNPs. To aid interpretability and recovery of a meaningful set of effects, we a perform hierarchical clustering with a minimax linkage on the set of $5233$ SNPs, see \cite{bien2011hierarchical}. The distance measure between SNPs $X_i$ and $X_j$ is defined as $d(X_i, X_j) = 1 -\rho(X_i,X_j)$ where $\rho(X_i,X_j)$ is the empirical correlation between two SNPs $X_i$, $X_j \in \real^{97}$. This algorithm introduced in \cite{ao2004clustag} clusters the SNPs and gives a prototype for each cluster. The number of clusters is chosen so that each of the $5233$ SNPs has a correlation of at least $0.5$ with  at least one of the prototypes. Applying a typical selection procedure like the Lasso on the set of local variants without pruning it to prototypes is not ideal in this analysis as the local variants share substantial empirical correlation; the Lasso will typically suffer from an inability to recover the true set of signals. \cite{reid2016sparse} identifies this shortcoming of the Lasso and proposes inference on effect sizes post a Lasso on prototypes of clusters in such scenarios. While the prototypes in \cite{reid2016sparse} are determined in a greedy fashion; the cluster representative being the most associated with the response, we adapt a completely unsupervised approach here in order to determine the clusters and prototypes with no data-snooping. Using the described hierarchical clustering, we obtain $320$ prototype SNPs, each of which has a correlation of at least $0.5$ with the SNPs in its cluster.
%The number of clusters is set at $320$ clusters with a cut-off for the minimax linkage tree determined such that the prototype SNPs have a correlation at least $0.5$ with any of the other SNPs in their respective clusters. 
We finally run a randomized Lasso query given by \eqref{lasso:randomized:program} on the prototype SNPs with Gaussian randomization. With the tuning parameter is set at the theoretical $\lambda = \hat{\sigma}\cdot\mathbb{E}(X^T\epsilon)$, the randomized Lasso query selects a set of $21$ potential regulatory variants; $\hat{\sigma}$ is estimated as $0.4$. The ratio of the randomization to noise scale in the data is set at $0.5$. 

We provide inference for the population least squares coefficients that correspond to the selected set $E$ of SNPs. That is, the adaptive target 
\[(X_E^T X_E)^{-1}X_E^T \mathbb{E}(Y\lvert X, \beta)\] is used as a quantification of the effect sizes of the selected SNPs. We assume the selected model on the data for inference, that is, $Y\sim \mathcal{N}(X_E\beta_E, \hat\sigma^2 I)$ and a non-informative prior on $\beta_E$ (similar to the simulations in Section \ref{simulated:models}). Figure \ref{egene:plot:0} gives a comparison of the effect sizes of selected SNPs using the proposed truncated approach with the unadjusted Bayesian estimates. Under the diffuse prior, the unadjusted estimates will be centered around the OLS estimator $(X_E^T X_E)^{-1}X_E^T y$ with variance given by the diagonal entries of $(X_E^T X_E)^{-1}$.  
The optimizations that we solve to obtain the truncated Bayesian estimates are laid out in Section \ref{sampler}.
We note the differences in effect sizes based on the adaptive posterior and the unadjusted posterior; we can see that the selected SNPs at positions $\bf{492, 606, 2960, 3509, 3574}$ will be reported as significantly associated with the gene expression if the analyst did not account for selection. The adjusted inference however, shows that the effect sizes of these SNPs are significantly biased by selection; the adjusted Bayesian intervals for these reportedly significant SNPs cover $0$.
\begin{figure}[h]
%\vskip 0.2in
\begin{center}
\centerline{\includegraphics[height=9cm,width=12cm]{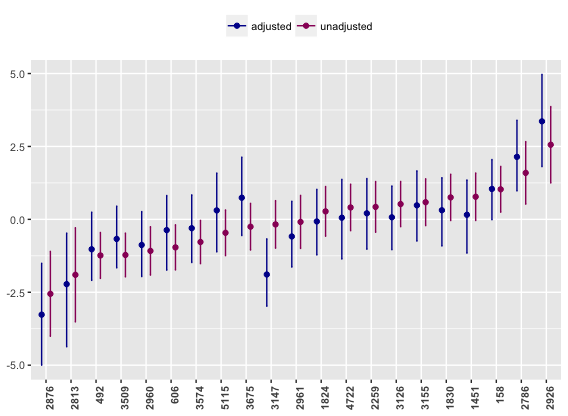}}
\caption{\small{Effect size estimates: posterior mean and credible intervals are based on the truncated and unadjusted Bayesian posterior. The adjusted intervals have an average length of 2.31, the unadjusted intervals have an average length of 1.86.}}
\label{egene:plot:0}
\end{center}
%\vskip -0.2in
%\vspace{-0.8cm}
\end{figure} 

We supplement the above randomized effect size estimates with non randomized frequentist inference of \cite{exact_lasso} post the usual Lasso query (without the randomization term in \eqref{lasso:randomized:program}). Figure \ref{egene:plot:1} plots the exact frequentist intervals post a Lasso selection compared against the unadjusted intervals. The non-randomized selection includes $18$ SNPs as opposed to $21$ SNPs picked up by the randomized Lasso; the common SNPs picked by both queries occur at positions $\bf{158, 492, 606, 1830, 2259, 2786, 2876, 2926, 2960, 3155,3509, 3574}$. The exact frequentist intervals adjusted for selection again show that the SNPs at $\bf{492, 606,  2960, 3509, 3574}$ are no longer statistically significant effects, as opposed to the unadjusted estimates. The comparison with the randomized intervals in Figure \ref{egene:plot:0} shows that the estimates post a randomized version of Lasso have more statistical power. This is highlighted in the shorter lengths of the randomized intervals when compared against the exact frequentist intervals of Figure \ref{egene:plot:1}. Adjusted inference post both randomized and non-randomized versions of the Lasso query identifies SNPs at $\bf{158, 2786, 2926}$ to be significantly associated with gene expression.
\begin{figure}[h]
%\vskip 0.2in
\begin{center}
\centerline{\includegraphics[height=9cm,width=12cm]{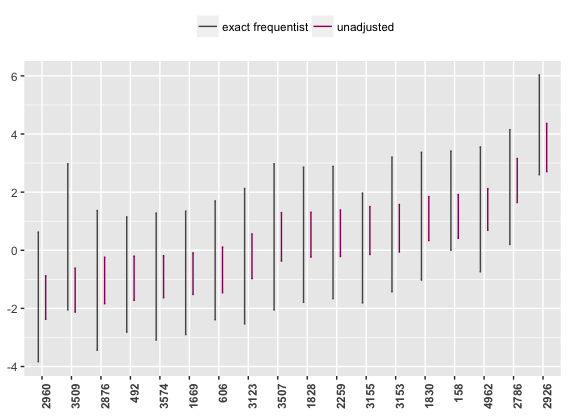}}
\caption{\small{Effect size estimates: adjusted intervals are the exact frequentist intervals constructed by conditioning on the polyhedral selection region of Lasso. Unadjusted intervals are centered around the OLS estimator post Lasso, $(X_E^T X_E)^{-1}X_E^T y$ with variance of $j$-th coefficient given by $(X_E^T X_E)^{-1}_{j,j}$. The adjusted intervals have an average length of 4.25, the unadjusted intervals have an average length of 1.56.}}
\label{egene:plot:1}
\end{center}
%\vskip -0.2in
%\vspace{-0.8cm}
\end{figure} 

The simulations in \ref{simulated:models} post different selective queries show that the Bayesian estimates have good frequentist properties under a non-informative prior. We show that the adjusted Bayesian estimates indeed mimic the adjusted frequentist estimates based on \cite{panigrahi2017mcmc} under the diffuse prior for the selected SNPs. Figure \ref{egene:plot:freq} in the Appendix \ref{egene:data:appendix} depicts the adjusted frequentist intervals alongside the Bayesian intervals.

To validate the inferential guarantees of our estimates in the above analysis, we conclude with a simulation design based on the predictor matrix of SNPs $X$ as considered above. We consider a sparse regime varying the number of signals $|S| \in \{0,1,2,3\}$. In this sparse regime, we simulate the response $Y$ from a model based on the $5233$ predictors with $|S|$ true signals as follows:
\begin{itemize}
\item subsample $|S|$ clusters from the $320$ clusters of SNPs (obtained by hierarchical clustering with a minimax linkage, described above)
\item subsample one SNP further from each of the $|S|$ subsampled cluster as the positions of the true signals; the set of true signals is called $\mathcal{S}$ with cardinality $|S|$.
\item draw response $y \in \real^{97}$ as $y = X_{\mathcal{S}}\beta_{\mathcal{S}} + \epsilon; \; \epsilon \sim \mathcal{N}(0, I),$
where $|\beta_{j,\mathcal{S}}|$ are of equal strength for $j \in \mathcal{S}$. We vary the the magnitude of signals over the set $\{10, 5, 2.5\}$ corresponding to roughly $K\sqrt{2\cdot \log p}$ with $K= 4.5, 2, 1;\; p=320$ respectively. These signal strengths correspond to three SNR regimes- strong, moderate and weak signal regime.
\end{itemize}
We evaluate the coverage and risk for the adjusted and unadjusted estimates averaged over the selected SNPs and across repetitions of an experiment with $50$ trials in the $3$ different signal regimes. In each repetition, we provide inference about the ground truth for the population least squares coefficients given by 
\[(X_E^T X_E)^{-1}X_E^T \mathbb{E}(Y\lvert X, \beta) = (X_E^T X_E)^{-1} X_E^T X_{\mathcal{S}}\beta_{\mathcal{S}}\]
under the selected model and non-informative prior as before. Note that the prototypes might not be positions of true signals in our simulation study, thereby, the model we assume for inference might be a misspecified model. We see that even with a misspecified model, the adjusted Bayesian estimates show superior performance than the unadjusted estimates, both in terms of coverage and risk. We also compare the adjusted Bayesian inference post the randomized Lasso with the non-randomized exact frequentist estimates of \cite{exact_lasso}. The blue color gives the adjusted Bayesian inference under the diffuse prior, the grey color represents the exact frequentist inference of \cite{exact_lasso} post a non-randomized Lasso query and the red color denotes the unadjusted Bayesian inference. \cite{exact_lasso} does not give a selection adjusted point estimate; the grey curve in Figure \ref{risk:egene} plots the risk of the OLS estimator $(X_E^T X_E)^{-1}X_E^Ty$ where $E$ is the set of SNPs selected by (non-randomized) Lasso. The labels in the x-axis represent the number of true signals, $|S|= 0,1,2,3$. The columns give comparison of estimates in the three signal regimes: from left to right- strong, moderate and weak corresponding to model of equally contributing signals with varying strengths $10, 5, 2.5$ respectively. 

\begin{figure}[H]
%\vspace{-2cm}
%\vskip 0.2in
\begin{center}
\centerline{\includegraphics[height=5cm,width=10cm]{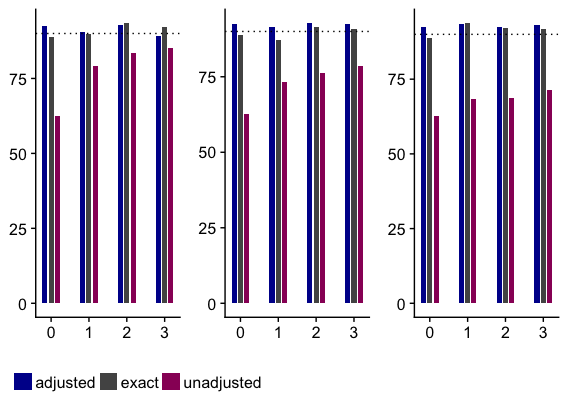}}
\caption{\small{Comparison of coverages in the strong, moderate and weak signal regime: the bar plot depicts the average coverages across 50 replications, averaged over the selected SNPs. The black dotted line marks the $90\%$ target coverage. The adjusted Bayesian intervals and exact \cite{exact_lasso} intervals cover the true target nearly $90\%$ of the total replications. The unadjusted intervals clearly fall short of the target coverage.}}
\label{coverage:egene}
\end{center}
%\vspace{-2cm}
%\vskip -0.2in
\end{figure} 
\begin{figure}[H]
\vspace{-1cm}
%\vskip 0.2in
\begin{center}
\centerline{\includegraphics[height=4.5cm,width=9cm]{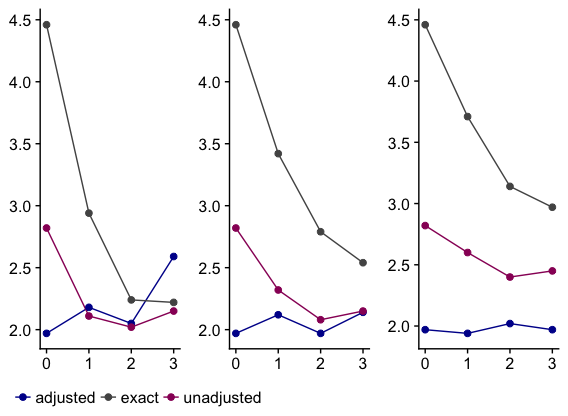}}
\caption{\small{Comparison of risk in strong, moderate and weak signal regime: the adjusted posterior mean has smaller risk than the inadaptive posterior mean post both randomized and non-randomized Lasso. The grey curve depicts the risk of the unadjusted posterior mean post non-randomized Lasso.}}
\label{risk:egene}
\end{center}
%\vspace{-2cm}
%\vskip -0.2in
\end{figure} 
\begin{figure}[H]
%\vspace{-1cm}
%\vskip 0.2in
\begin{center}
\centerline{\includegraphics[height=4.5cm,width=9cm]{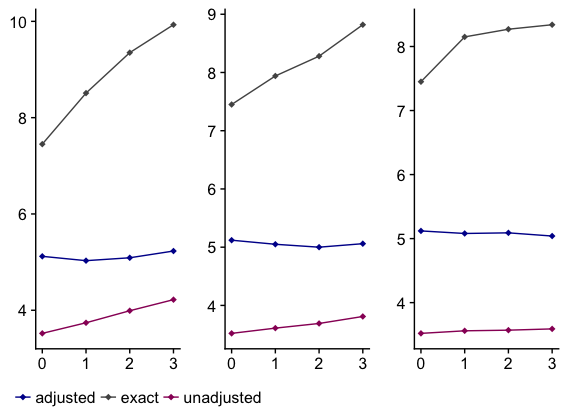}}
\caption{\small{Comparison of lengths in strong, moderate and weak signal regime: the adjusted Bayesian intervals are much shorter than the exact frequentist intervals, they are comparable in length to the unadjusted intervals.}}
\label{lengths:egene}
\end{center}
%\vspace{-1cm}
%\vskip -0.2in
\end{figure} 
\section{Concluding remarks}
The motivations to adjust for selection through a truncation on the generative model is the same as the frequentist line of works; though the technicalities with imposing a Bayesian model post selection are different. While prior works make progress in formalizing a selective Bayesian methodology, the current work makes significant contributions in proposing a concrete computational recipe to approximate the selective posterior after systematic randomized procedures. The methods extend to multi stage selective queries, marginalizing over randomizations from each selective stage. An attractive property of this approach is scalability in both regimes of inference with empirical demonstration of frequentist coverage with credible intervals and risk of the posterior mean based on the approximate selective posterior. 
%The main take away is that the analyst can take usual full advantage of the Bayesian machinery once he has the tools to sample from the truncated posterior; this is the contribution of the paper.

An interesting future direction includes establishing a Bernstein von Mises result in the selective Bayesian paradigm. We empirically see that the truncated Bayesian methods somewhat recover frequentist coverage under diffuse priors just as they would in the untruncated regime of inference. From a purely application point of view, we see scope of the methodology in this work to be applied to genome-wide studies where the true model describing the association of phenotypes with variants is assumed to be highly sparse. Inference post identification of causal variants is an important goal; our methods can provide reliable and reproducible  effect size estimates in such settings. The Bayesian model in particular, allows an analyst to leverage information from an objective or subjective prior that can arise from prior experimentation in these studies. Developing tools to sample from an intractable posterior modeled using the truncated framework has been the focus of this paper, this allows the analyst to take full advantage of the Bayesian machinery post selection and provide estimates with better coverage and risk properties than the usual Bayesian estimates.\\

\noindent \textbf{Acknowledgement} 
The data used for the gene expression data analysis described in this manuscript was obtained using dbGaP accession number \textbf{phs000424.v6.p1}. The authors are extremely thankful to Chiara Sabbati for her help in acquisition of the GTEx gene-expression data in this work. The authors would like to thank Asaf Weinstein, with whom they collaborated on formalizing many ideas in \cite{selective_bayesian}. The authors also acknowledge helpful discussions with Chiara Sabbati and Robert Tibshirani which has improved their understanding of the problem. 
%%%%%%%%%%%%%%%%%%%%%%%%%%%%%%%%%%%%%%%%%%%%%%%%%%%%
%%%%%%%%%%%%%%%%%%%%%%%%%%%%%%%%%%%%%%%%%%%%%%%%%%%%

%\clearpage
\bibliographystyle{plainnat}
\bibliography{references} 

\begin{thebibliography}{33}
\providecommand{\natexlab}[1]{#1}
\providecommand{\url}[1]{\texttt{#1}}
\expandafter\ifx\csname urlstyle\endcsname\relax
  \providecommand{\doi}[1]{doi: #1}\else
  \providecommand{\doi}{doi: \begingroup \urlstyle{rm}\Url}\fi

\bibitem[Aguet et~al.(2016)Aguet, Brown, Castel, Davis, Mohammadi, Segre,
  Zappala, Abell, Fresard, Gamazon, et~al.]{aguet2016local}
Francois Aguet, Andrew~A Brown, Stephane Castel, Joe~R Davis, Pejman Mohammadi,
  Ayellet~V Segre, Zachary Zappala, Nathan~S Abell, Laure Fresard, Eric~R
  Gamazon, et~al.
\newblock Local genetic effects on gene expression across 44 human tissues.
\newblock \emph{BiorXiv}, page 074450, 2016.

\bibitem[Ao et~al.(2004)Ao, Yip, Ng, Cheung, Fong, Melhado, and
  Sham]{ao2004clustag}
Sio~Iong Ao, Kevin Yip, Michael Ng, David Cheung, Pui-Yee Fong, Ian Melhado,
  and Pak~C Sham.
\newblock Clustag: hierarchical clustering and graph methods for selecting tag
  snps.
\newblock \emph{Bioinformatics}, 21\penalty0 (8):\penalty0 1735--1736, 2004.

\bibitem[Bien and Tibshirani(2011)]{bien2011hierarchical}
Jacob Bien and Robert Tibshirani.
\newblock Hierarchical clustering with prototypes via minimax linkage.
\newblock \emph{Journal of the American Statistical Association}, 106\penalty0
  (495):\penalty0 1075--1084, 2011.

\bibitem[Carithers et~al.(2015)Carithers, Ardlie, Barcus, Branton, Britton,
  Buia, Compton, DeLuca, Peter-Demchok, Gelfand, et~al.]{carithers2015novel}
Latarsha~J Carithers, Kristin Ardlie, Mary Barcus, Philip~A Branton, Angela
  Britton, Stephen~A Buia, Carolyn~C Compton, David~S DeLuca, Joanne
  Peter-Demchok, Ellen~T Gelfand, et~al.
\newblock A novel approach to high-quality postmortem tissue procurement: the
  gtex project.
\newblock \emph{Biopreservation and biobanking}, 13\penalty0 (5):\penalty0
  311--319, 2015.

\bibitem[Chaudhuri and Monteleoni(2009)]{chaudhuri2009privacy}
Kamalika Chaudhuri and Claire Monteleoni.
\newblock Privacy-preserving logistic regression.
\newblock In \emph{Advances in Neural Information Processing Systems}, pages
  289--296, 2009.

\bibitem[Chaudhuri et~al.(2011)Chaudhuri, Monteleoni, and
  Sarwate]{chaudhuri2011differentially}
Kamalika Chaudhuri, Claire Monteleoni, and Anand~D Sarwate.
\newblock Differentially private empirical risk minimization.
\newblock \emph{Journal of Machine Learning Research}, 12\penalty0
  (Mar):\penalty0 1069--1109, 2011.

\bibitem[Consortium et~al.(2015)]{gtex2015genotype}
GTEx Consortium et~al.
\newblock The genotype-tissue expression (gtex) pilot analysis: Multitissue
  gene regulation in humans.
\newblock \emph{Science}, 348\penalty0 (6235):\penalty0 648--660, 2015.

\bibitem[Dembo and Zeitouni(1998)]{dembo1998large}
Amir Dembo and Ofer Zeitouni.
\newblock Large deviations techniques and applications second edition.
\newblock \emph{Large deviations techniques and applications}, 38, 1998.

\bibitem[Dwork et~al.(2015)Dwork, Feldman, Hardt, Pitassi, Reingold, and
  Roth]{dwork2015preserving}
Cynthia Dwork, Vitaly Feldman, Moritz Hardt, Toniann Pitassi, Omer Reingold,
  and Aaron~Leon Roth.
\newblock Preserving statistical validity in adaptive data analysis.
\newblock In \emph{Proceedings of the Forty-Seventh Annual ACM on Symposium on
  Theory of Computing}, pages 117--126. ACM, 2015.

\bibitem[Fithian et~al.(2014)Fithian, Sun, and Taylor]{optimal_inference}
William Fithian, Dennis Sun, and Jonathan Taylor.
\newblock Optimal {Inference} {After} {Model} {Selection}.
\newblock \emph{arXiv preprint arXiv:1410.2597}, October 2014.
\newblock URL \url{http://arxiv.org/abs/1410.2597}.
\newblock arXiv: 1410.2597.

\bibitem[George and McCulloch(1997)]{george1997approaches}
Edward~I George and Robert~E McCulloch.
\newblock Approaches for bayesian variable selection.
\newblock \emph{Statistica sinica}, pages 339--373, 1997.

\bibitem[Hoffman et~al.(2013)Hoffman, Blei, Wang, and
  Paisley]{hoffman2013stochastic}
Matthew~D Hoffman, David~M Blei, Chong Wang, and John~William Paisley.
\newblock Stochastic variational inference.
\newblock \emph{Journal of Machine Learning Research}, 14\penalty0
  (1):\penalty0 1303--1347, 2013.

\bibitem[Lee et~al.(2016)Lee, Sun, Sun, and Taylor]{exact_lasso}
Jason~D. Lee, Dennis~L. Sun, Yuekai Sun, and Jonathan~E. Taylor.
\newblock Exact post-selection inference with the lasso.
\newblock \emph{The Annals of Statistics}, 44\penalty0 (3):\penalty0 907--927,
  November 2016.
\newblock URL \url{http://projecteuclid.org/euclid.aos/1460381681}.

\bibitem[Loftus and Taylor(2014)]{loftus_quadratic}
Joshua~R. Loftus and Jonathan~E. Taylor.
\newblock A significance test for forward stepwise model selection.
\newblock May 2014.
\newblock URL \url{http://xxx.tau.ac.il/abs/1405.3920v1}.

\bibitem[Loftus and Taylor(2015)]{loftus2015selective}
Joshua~R Loftus and Jonathan~E Taylor.
\newblock Selective inference in regression models with groups of variables.
\newblock \emph{arXiv preprint arXiv:1511.01478}, 2015.

\bibitem[Markovic and Taylor(2016)]{markovic2016bootstrap}
Jelena Markovic and Jonathan Taylor.
\newblock Bootstrap inference after using multiple queries for model selection.
\newblock \emph{arXiv preprint arXiv:1612.07811}, 2016.

\bibitem[Minka(2001)]{minka2001family}
Thomas~P Minka.
\newblock \emph{A family of algorithms for approximate Bayesian inference}.
\newblock PhD thesis, Massachusetts Institute of Technology, 2001.

\bibitem[Mitchell and Beauchamp(1988)]{mitchell1988bayesian}
Toby~J Mitchell and John~J Beauchamp.
\newblock Bayesian variable selection in linear regression.
\newblock \emph{Journal of the American Statistical Association}, 83\penalty0
  (404):\penalty0 1023--1032, 1988.

\bibitem[Negahban et~al.(2009)Negahban, Yu, Wainwright, and
  Ravikumar]{negahban2009unified}
Sahand Negahban, Bin Yu, Martin~J Wainwright, and Pradeep~K Ravikumar.
\newblock A unified framework for high-dimensional analysis of $ m $-estimators
  with decomposable regularizers.
\newblock In \emph{Advances in Neural Information Processing Systems}, pages
  1348--1356, 2009.

\bibitem[Ongen et~al.(2015)Ongen, Buil, Brown, Dermitzakis, and
  Delaneau]{ongen2015fast}
Halit Ongen, Alfonso Buil, Andrew~Anand Brown, Emmanouil~T Dermitzakis, and
  Olivier Delaneau.
\newblock Fast and efficient qtl mapper for thousands of molecular phenotypes.
\newblock \emph{Bioinformatics}, 32\penalty0 (10):\penalty0 1479--1485, 2015.

\bibitem[Panigrahi et~al.(2016)Panigrahi, Taylor, and
  Weinstein]{selective_bayesian}
Snigdha Panigrahi, Jonathan Taylor, and Asaf Weinstein.
\newblock Bayesian post-selection inference in the linear model.
\newblock \emph{arXiv preprint arXiv:1605.08824}, 2016.
\newblock URL \url{https://arxiv.org/abs/1605.08824}.

\bibitem[Panigrahi et~al.(2017)Panigrahi, Markovic, and
  Taylor]{panigrahi2017mcmc}
Snigdha Panigrahi, Jelena Markovic, and Jonathan Taylor.
\newblock An mcmc free approach to post-selective inference.
\newblock \emph{arXiv preprint arXiv:1703.06154}, 2017.

\bibitem[Reid and Tibshirani(2016)]{reid2016sparse}
Stephen Reid and Robert Tibshirani.
\newblock Sparse regression and marginal testing using cluster prototypes.
\newblock \emph{Biostatistics}, 17\penalty0 (2):\penalty0 364--376, 2016.

\bibitem[Roberts and Tweedie(1996)]{roberts1996exponential}
Gareth~O Roberts and Richard~L Tweedie.
\newblock Exponential convergence of langevin distributions and their discrete
  approximations.
\newblock \emph{Bernoulli}, pages 341--363, 1996.

\bibitem[Taylor and Tibshirani(2016)]{penalized_l1}
Jonathan Taylor and Robert Tibshirani.
\newblock Post-selection inference for l1-penalized likelihood models.
\newblock \emph{arXiv preprint arXiv:1602.07358}, 2016.
\newblock URL \url{http://arxiv.org/abs/1602.07358}.

\bibitem[Taylor et~al.(2013)Taylor, Loftus, and Tibshirani]{kac_rice}
Jonathan Taylor, Joshua Loftus, and Ryan Tibshirani.
\newblock Tests in adaptive regression via the {Kac}-{Rice} formula.
\newblock \emph{The Annals of Statistics}, 44\penalty0 (2):\penalty0 743--770,
  August 2013.
\newblock URL \url{http://projecteuclid.org/euclid.aos/1458245734}.

\bibitem[Tian and Taylor(2015)]{randomized_response}
Xiaoying Tian and Jonathan~E. Taylor.
\newblock Selective inference with a randomized response.
\newblock \emph{arXiv preprint arXiv:1507.06739}, July 2015.
\newblock URL \url{http://arxiv.org/abs/1507.06739}.
\newblock arXiv: 1507.06739.

\bibitem[Tian et~al.(2016)Tian, Panigrahi, Markovic, Bi, and
  Taylor]{harris2016selective}
Xiaoying Tian, Snigdha Panigrahi, Jelena Markovic, Nan Bi, and Jonathan Taylor.
\newblock Selective sampling after solving a convex problem.
\newblock \emph{arXiv preprint arXiv:1609.05609}, 2016.

\bibitem[Tibshirani et~al.(2014)Tibshirani, Taylor, Lockhart, and
  Tibshirani]{spacings}
Ryan Tibshirani, Jonathan Taylor, Richard Lockhart, and Robert Tibshirani.
\newblock Post-selection adaptive inference for {Least} {Angle} {Regression}
  and the {Lasso}.
\newblock \emph{arXiv preprint arXiv:1401.3889}, January 2014.
\newblock URL \url{http://arxiv.org/abs/1401.3889}.

\bibitem[Tibshirani et~al.(2016)Tibshirani, Taylor, Lockhart, and
  Tibshirani]{tibshirani2016exact}
Ryan~J Tibshirani, Jonathan Taylor, Richard Lockhart, and Robert Tibshirani.
\newblock Exact post-selection inference for sequential regression procedures.
\newblock \emph{Journal of the American Statistical Association}, 111\penalty0
  (514):\penalty0 600--620, 2016.

\bibitem[Yang et~al.(2016)Yang, Barber, Jain, and Lafferty]{yang2016selective}
Fan Yang, Rina~Foygel Barber, Prateek Jain, and John Lafferty.
\newblock Selective inference for group-sparse linear models.
\newblock In \emph{Advances in Neural Information Processing Systems}, pages
  2469--2477, 2016.

\bibitem[Yekutieli(2012)]{yekutieli2012adjusted}
Daniel Yekutieli.
\newblock Adjusted bayesian inference for selected parameters.
\newblock \emph{Journal of the Royal Statistical Society: Series B (Statistical
  Methodology)}, 74\penalty0 (3):\penalty0 515--541, 2012.

\bibitem[Zou and Hastie(2005)]{elastic_net}
Hui Zou and Trevor Hastie.
\newblock Regularization and variable selection via the elastic net.
\newblock \emph{Journal of the Royal Statistical Society: Series B},
  67\penalty0 (2):\penalty0 301--320, 2005.
\newblock URL
  \url{http://onlinelibrary.wiley.com/doi/10.1111/j.\\1467-9868.2005.00503.x/abstract}.

\end{thebibliography}

\appendix
%--------------------------------------------------
\section{Proofs of Theorems}
\label{appendix}
%--------------------------------------------------

\subsection{Results in Sections \ref{optimization} and \ref{sampler}}
\label{A:1}
\noindent{\textbf{Proof of Theorem \ref{approximate:prob}}}:
\begin{proof} 
%Using the Chernoff bound in \eqref{Chernoff:bound} 
To prove this, we derive an upper bound on $\log {\mathbb{P}}((S,O) \in \mathcal{R}\lvert \beta^*)$ in terms of the log-MGF of the augmented random variable $(S,O)$.
\begin{align}
&\log {\mathbb{P}}((S,O) \in \mathcal{R}\lvert \beta^*)\nonumber\\
& \leq \log \mathbb{E}\left[\exp(\sup_{s, o\in \mathcal{R}}\{-\alpha_1^T s -\alpha_2^T o\})\exp(\alpha_1^T S + \alpha_2^T O)\Big\lvert \beta^*\right ]\nonumber\\
&=- \inf_{s, o\in \mathcal{R}}\{\alpha_1^T s + \alpha_2^T o \} -\log \mathbb{E}[\exp(\alpha_1^T S + \alpha_2^T O)\lvert \beta^*] \nonumber
\end{align}
Since the above bound holds for any $\alpha_1\in \real^d$ and $\alpha_2 \in \real^p$, we can optimize over the choices of $\alpha_1, \alpha_2$ to obtain the upper bound 
\begin{equation*}
-\sup_{\alpha_1, \alpha_2} \inf_{s, o\in \mathcal{R}}\{\alpha_1^T s + \alpha_2^T o \} -\log \mathbb{E}[\exp(\alpha_1^T S + \alpha_2^T O)\lvert \beta^*].
\end{equation*}
A minimax equality for convex, compact selection region $\mathcal{R}$ justifies the swapping of infimum and supremum to lead to the bound
\begin{equation}
\label{min:max}
%\begin{aligned}
%&
{-\inf\limits_{s,o \in \mathcal{R}}\Big\{\sup\limits_{\alpha_1, \alpha_2}\alpha_1^T s + \alpha_2^T o - \log\mathbb{E}(\exp(\alpha_1^T S + \alpha_2^T O)\lvert \beta^*) \Big\}.}
%\end{aligned}
\end{equation}
The main step is computation of the log-MGF $\log\mathbb{E}(\exp(\alpha_1^T S + \alpha_2^T O)\lvert \beta^*)$, which is possible through the change of measure facilitated by the inversion map in \eqref{randomization:map}.
Using the joint density of the vector $(S,O)$ in \eqref{selective:density} and writing
\[o= P^{-1}(\omega-Ds-q),\]
we have
%to be able to calculate $ \log\mathbb{E}(\exp(\alpha_1^T S + \alpha_2^T O)\lvert \beta^*)$ in terms of the log-MGF corresponding to the data $ \Lambda_f(.\lvert \beta^*) $ and the randomization  $\Lambda_g(.)$.
%\begin{equation*}
\begin{align}
\allowdisplaybreaks
&{\E[\exp(\alpha_1^T S + \alpha_2^T O)\lvert \beta^*] }\nonumber\\
&= {\int \exp\left(\alpha_1^T s+ \alpha_2^T o\right) |J| f(s\lvert \beta^*)g(Ds + Po + q)ds do}\nonumber\\
&= {\int \exp\left(\alpha_1^T s+ \alpha_2^T P^{-1}(\omega-Ds-q)\right) f(s\lvert \beta^*)g(\omega)ds d\omega}\nonumber\\
&=  {\exp( -\alpha_2^T P^{-1} q)\int \exp((\alpha_1-D^T {P^{-1}}^T \alpha_2)^T s) f(s\lvert \beta^*)ds}\nonumber\\
&  {\;\;\;\;\;\times \int \exp( \alpha_2^T P^{-1}\omega)g(\omega)d\omega} \nonumber\\
&=  {\exp( -\tilde{\alpha}_2^T q) \E[\exp(\tilde{\alpha}_1 S)\lvert \beta^*)] \cdot\E[\exp(\tilde{\alpha}_2^T \Omega)] }\nonumber\\
&=   {\exp( -\tilde{\alpha}_2^T q)\exp\Lambda_f(\tilde{\alpha}_1\lvert \beta^*)\exp\Lambda_g(\tilde{\alpha}_2)}\nonumber
\end{align}
%\end{equation*}
with $\tilde{\alpha}_1=\alpha_1 -D^T (P^{-1})^T\alpha_2\text{ and }\tilde{\alpha}_2=(P^{-1})^T\alpha_2.$

\noindent Plugging
\begin{equation*}
\label{mgf:data:opt}
{\log \E[\exp(\alpha_1^T S + \alpha_2^T O \lvert \beta^*)]= -\tilde{\alpha}_2^T q +\Lambda_f(\tilde{\alpha}_1\lvert \beta^*)+\Lambda_g(\tilde{\alpha}_2)}
\end{equation*}
into \eqref{min:max} gives the upper bound for $\log {\mathbb{P}}((S,O) \in \mathcal{R}\lvert \beta^*)$ in terms of the log-MGF corresponding to the data $ \Lambda_f(.\lvert \beta^*) $ and the randomization  $\Lambda_g(.)$ as
%\begin{equation*}
%\label{approx:conjugate}
\begin{align}
%&\Scale[0.90]{\log \mathbb{P}((S,O) \in \mathcal{R}\lvert \beta^*) } \nonumber\\
& {-\inf\limits_{s, o\in \mathcal{R}}\Big\{\sup\limits_{\tilde{\alpha}_1,\tilde{\alpha}_2}\tilde{\alpha}_1^T s + \tilde{\alpha}_2^T (P o+ D s)- \log\E[\exp(\alpha_1^T S + \alpha_2^T O \lvert \beta^*)]\Big\}}\nonumber\\
%&\;\;\;\;\;\;\;\;\;\;\;\;\;\;\;\;\;\;\;\;- \log(\E[\exp(\alpha_1^T S + \alpha_2^T O \lvert \beta^*)]))\Big\} \nonumber\\
&={-\inf\limits_{s,o\in \mathcal{R}}\Big\{\sup\limits_{\tilde{\alpha}_1,\tilde{\alpha}_2}(\tilde{\alpha}_1^T s + \tilde{\alpha}_2^T (P o +D s+ q)- \Lambda_f(\tilde{\alpha}_1\lvert \beta^*)-\Lambda_g(\tilde{\alpha}_2))\Big\}}\nonumber\\
%&\;\;\;\;\;\;\;\;\;\;\;\;\;\;\;\;\;\;\;\;\;- \Lambda_f(\tilde{\alpha}_1\lvert \beta^*)-\Lambda_g(\tilde{\alpha}_2))\Big\}\nonumber\\
&= {-\inf_{s,o \in \mathcal{R}}  \Lambda_f^*(s\lvert \beta^*) +  \Lambda_g^*(Ds + P o + q)}.\nonumber
\end{align}
\end{proof}

\noindent{\textbf{Proof of Theorem \ref{bound:modified}}}:
\begin{proof}
The volume of the selection region 
\[\mathcal{R} = \mathcal{R}_S \times \mathcal{R}_O=\mathcal{R}_S \times  \mathcal{R}_E\times \prod_{j=1}^{p-|E|} \mathcal{R}_{j,-E}\]
based on decoupling of the randomization density into active and inactive coordinates is given by
\begin{equation*}
\begin{aligned}
& \mathbb{P}((S,O)\in \mathcal{R}\lvert \beta^*) \nonumber \\
&= |J|\cdot\int_{\mathcal{R}_S}f(s\lvert \beta^*)\cdot \int_{\mathcal{R}_E} g_{E}(D_E s+ P_E o_E + q_E)\nonumber \\
&\;\;\;\;\; \cdot  \prod_{j=1}^{p-|E|}\int_{\mathcal{R}_{j,-E}}g_{j,-E}(o_{j,-E} +D_{j,-E}s + P_{j,-E} o_{E}+ q_{j,-E})do_{j,-E} do_{E} ds \nonumber \\
&= |J| \cdot \int_{\mathcal{R}_S}f(s\lvert \beta^*)\cdot \int_{\mathcal{R}_E} g_{E}(D_E s+ P_E o_E + q_E)\mathcal{B}(o_E;s) do_E ds \nonumber \\
&= \mathbb{E}\left[\mathcal{B}(O_E;S)1_{(S,O_E)\in \mathcal{R}_S \times  \mathcal{R}_E}\right \lvert \beta^*]  \nonumber \\
\end{aligned}
\end{equation*}

\noindent An upper bound on $\log\mathbb{E}\left[\mathcal{B}(O_E;S)1_{(S,O_E)\in \mathcal{R}_S \times  \mathcal{R}_E}\right \lvert \beta^*] $ is given by
\begin{equation*}
\setlength{\jot}{1.2em}
\begin{aligned}
& \log \mathbb{E}\left[\mathcal{B}(O_E;S)1_{(S,O_E)\in \mathcal{R}_S \times  \mathcal{R}_E}\right \lvert \beta^*] \nonumber \\
&=\log  \mathbb{E}\left[\exp(\log\mathcal{B}(O_E;S)-\alpha_1^T S-\alpha_2^T O_E)\cdot \exp(\alpha_1^T S+\alpha_2^T O_E)\cdot 1_{(S,O_E)\in \mathcal{R}_S \times  \mathcal{R}_E}\right \lvert \beta^*] \nonumber \\
&\leq \sup_{s\in \mathcal{R}_S, o_E \in \mathcal{R}_E}\left\{\log\mathcal{B}(o_E;s)-\alpha_1^T s-\alpha_2^T o_E\right\} + \log\mathbb{E}\left[ \exp(\alpha_1^T S+\alpha_2^T O_E)\right \lvert \beta^*] \nonumber 
\end{aligned}
\end{equation*}
Optimizing over $\alpha_1\in \real^d$ and $\alpha_2\in \real^p$, we have 
\begin{equation*}
\setlength{\jot}{1.2em}
\begin{aligned}
& \log \mathbb{E}\left[\mathcal{B}(O_E;S)1_{(S,O_E)\in \mathcal{R}_S \times  \mathcal{R}_E}\right \lvert \beta^*] \nonumber \\
&\leq \inf_{\alpha_1, \alpha_2}\sup_{s\in\mathcal{R}_S, o_E \in \mathcal{R}_E}\left\{\log\mathcal{B}(o_E;s)-\alpha_1^T s-\alpha_2^T o_E\right\} + \log\mathbb{E}\left[ \exp(\alpha_1^T S+\alpha_2^T O_E)\right \lvert \beta^*] \nonumber \\
&=  \Scale[.95]{-\sup\limits_{\alpha_1, \alpha_2}\Big\{ \inf\limits_{s\in \mathcal{R}_S, o_E \in \mathcal{R}_E}\left\{\alpha_1^T s+ \alpha_2^T o_E-\log\mathcal{B}(o_E;s)\right\} - \log\mathbb{E}\left[ \exp(\alpha_1^T S+\alpha_2^T O_E)\right \lvert \beta^*] \Big\}}.
\end{aligned}
\end{equation*}
By a minimax equality for compact, convex selection region $\mathcal{R}_S \times \mathcal{R}_E$ and the expression for log-MGF using the change of measure derived in the proof of Theorem \ref{approximate:prob}, we have the result.
\end{proof}

\noindent{\textbf{Proof of Theorem \ref{dual:opt:appprox}:}}
\begin{proof}
With the introduction of variable $v= Ds + Po + q$, the dual of optimization 
\[{ \inf_{s,o} \left\{ \Lambda_f^*(s\lvert \beta^*) +  \Lambda_g^*(Ds + P o + q)+ b_{\mathcal{R}_O}(o)\right\}}\] in terms of dual variable $u\in \mathbb{R}^{p}$
\begin{equation*}
\begin{aligned}
{\sup\limits_{u}\inf\limits_{s,o,v}\Big\{ \Lambda_f^*(s\lvert \beta^*) +  \Lambda_g^*(v) +  b_{\mathcal{R}_O}(o) + u^T (v -Ds - P o - q)\Big\}}
\end{aligned}
\end{equation*}
Solving Lagrangian $\mathcal{L}(u)$ over variables $(s,o,v)$
%\begin{equation*}
\begin{align}
{ \inf_{s,o,v}\Big\{ \Lambda_f^*(s\lvert \beta^*) +  \Lambda_g^*(v) + b_{\mathcal{R}_O}(o)+ u^T (v -Ds - P o - q)\Big\}}\nonumber
\end{align}
%\end{equation*}
%yields optimized value $\mathcal{L}(u)= -\Lambda_f(D^T u\lvert \beta^*)-  \Lambda_g(-u)- b^*_{\mathcal{R}_O}(P^T u) -u^T q$
gives the optimizing equations
\begin{equation*}
s = \grad \Lambda_f(D^T u\lvert \beta^*); v = \grad \Lambda_g(- u),o = \grad {b^*_{\mathcal{R}_O}}^{-1}(P^{T} u).
\end{equation*}
This yields 
\[\mathcal{L}(u)= -\Lambda_f(D^T u\lvert \beta^*)-  \Lambda_g(-u)- b^*_{\mathcal{R}_O}(P^T u) -u^T q\]
and hence, follows \eqref{dual:barrier}.
\end{proof}

\noindent{\textbf{Proof of Theorem \ref{gradient:posterior}}:}
\begin{proof}
Plugging in the conjugate of the log-Gaussian MGF of data vector $S$, the logarithm of the pseudo posterior can be written as
\begin{equation*}
\begin{aligned}
\log \tilde{\pi}_E(\beta^*\lvert s) &= K + \log \pi(\beta^*) - \frac{(s-\mu(\beta^*))^T\Sigma_f^{-1}(s-\mu(\beta^*))}{2} \\
&-\sup_{z\in \real^d}\left\{ z^T\Sigma_f^{-1}\mu(\beta^*) -\frac{1}{2}z^T \Sigma_f^{-1}\ z - \inf_{o\in \mathbb{R}^p}\left\{ \Lambda_g^*(Dz + Po +q) + b_{\mathcal{R}_O}(o)\right\}\right\}\\
&= K + \log \pi(\beta^*) - \frac{(s-\mu(\beta^*))^T\Sigma_f^{-1}(s-\mu(\beta^*))}{2} - \delta^*(\Sigma_f^{-1}\mu(\beta^*)) \\
&\;\;\;+ \frac{\mu(\beta^*)^T\Sigma_f^{-1}\mu(\beta^*)}{2}
\end{aligned}
\end{equation*}
for constant $K=-d \log 2\pi/2 - \log |\Sigma_f|/2$ and $\delta^*$ representing the conjugate of 
\[\delta(z) = \frac{1}{2}z^T \Sigma_f^{-1}\ z + \inf_{o\in \mathbb{R}^p}\left\{ \Lambda_g^*(Dz + Po +q) + b_{\mathcal{R}_O}(o)\right\}.\]
 The derivative of the log-pseudo posterior is thus given by
\begin{equation*}
\begin{aligned}
 \cfrac{\partial \log \tilde{\pi}_E(\beta^*\lvert s)}{\partial \beta^*} &= \cfrac{\partial\log\pi(\beta^*)}{\partial \beta^*} + \left(\cfrac{\partial\mu}{\partial \beta^*}\right)^T \Sigma_f^{-1}s-\left(\cfrac{\partial\mu}{\partial \beta^*}\right)^T \Sigma_f^{-1}\grad \delta^*(  \Sigma_f^{-1}\mu)\\
&= \cfrac{\partial\log\pi(\beta^*)}{\partial \beta^*} + \left(\cfrac{\partial\mu}{\partial \beta^*}\right)^T \Sigma_f^{-1}(s-s^*(\Sigma_f^{-1}\mu(\beta^*)) \\
\end{aligned}
\end{equation*}
for $s^*$ satisfying
\[\arg\inf_{z\in \mathbb{R}^d}\left( z^T\Sigma_f^{-1}\mu(\beta^*) -\frac{1}{2}z^T \Sigma_f^{-1}\ z - \inf_{o\in \mathbb{R}^p}\left\{ \Lambda_g^*(Dz + Po +q) + b_{\mathcal{R}_O}(o)\right\}\right ).\]
The last equality follows by noting that 
\[\grad \delta^*(  \Sigma_f^{-1}\mu) = \grad \delta^{-1}(  \Sigma_f^{-1}\mu) = s^*( \Sigma_f^{-1}\mu).\]
\end{proof}

\subsection{Results in Section \ref{LDP}}
\label{A:3}
\noindent{\textbf{Proof of Theorem \ref{LDP:thm}}}
\begin{proof} 
\begin{enumerate}[(1).]
\item 
We can write $\lim_n \frac{1}{n}\log \mathbb{P}(\bar{S}_n  \in \mathcal{R}_S^{'}, \bar{O}_n \in \mathcal{R}_O^{'}\lvert \beta_E)$ as
\begin{equation*}
\begin{aligned}
& \lim_n \frac{1}{n}\log \mathbb{E}\left[\mathbb{P}(\bar{O}_n \in \mathcal{R}_O^{'} \lvert \bar{S}_n=s)1_{\bar{S}_n \in  \mathcal{R}_S}\lvert \beta_E \right]\\
&= \lim_n \frac{1}{n}\log \mathbb{E}\left[\exp\left(n \cdot \frac{1}{n} \log \mathbb{P}(\bar{O}_n \in \mathcal{R}_O^{'} \lvert \bar{S}_n=s)\right)1_{\bar{S}_n \in  \mathcal{R}_S^{'}}\Big\lvert \beta_E \right].
\end{aligned}
\end{equation*}
Note that  
\[H_n(s)= \frac{1}{n} \log \mathbb{P}(\bar{O}_n \in \mathcal{R}_O^{'} \lvert \bar{S}_n=s)\]
satisfies the limit \eqref{LDP:rate:optimization}. We also use the observation that 
\[H^{'}_n(s) = -\inf_{o\in \mathcal{R}_O^{'}} \Lambda_g^*(D s+ Po+ q/\sqrt{n}) \]
satisfies
$\lim_n H^{'}_n(s)  = -\inf_{o\in \mathcal{R}_O^{'}} \Lambda_g^*(D s+ Po);$
this follows as the limiting sequence of objectives, composition of an affine map with the conjugate of log-Gaussian MGF, is convex and converges to a non-monotonic convex objective $\Lambda_g^*(D s+ Po)$. These two facts ensure that $H_n(.)$ and $H^{'}_n(.)$ are two sequences of continuous functions that converge uniformly on the set $\mathcal{R}_S$ to the limit $H(s) = -\inf_{o\in \mathbb{R}_O^{'}} \Lambda_g^*(D s+ Po)$. A direct application of Lemma \ref{varadhan:lemma} now leads to the limiting rate
\begin{equation*}
\begin{aligned}
&\lim_{n}\frac{1}{n} \log \mathbb{E}\left [\exp(n H_n(\bar{S}_n))1_{\bar{S}_n \in  \mathcal{R}_S}\lvert \beta_E\right ]\\
&=  \lim_{n}\frac{1}{n} \log \mathbb{E}\left [\exp(n H^{'}_n(\bar{S}_n))1_{\bar{S}_n \in  \mathcal{R}_S}\lvert \beta_E\right ]\\
&= -\inf_{s\in \mathcal{R}_S^{'}, o \in \mathcal{R}_O^{'}} \{\Lambda_f^*(s\lvert \beta_E) + \Lambda_g^*(Ds + Po)\}\\
&= -\lim_n\inf_{s\in \mathcal{R}_S^{'}, o \in \mathcal{R}_O^{'}} \{\Lambda_f^*(s\lvert \beta_E) + \Lambda_g^*(Ds + Po +q/\sqrt{n}) \}.
\end{aligned}
\end{equation*}
\item For the second part, we use the tower property of expectation to have 
\[\Scale[0.96]{\mathbb{P}\left(\bar{O}_n \in \mathcal{R}_O^{'} \lvert \bar{S}_n=s\right) = \mathbb{E}\left[\mathbb{P}(\bar{O}_{-E,n} \in \mathcal{R}_{-E}^{'}\lvert  \bar{O}_{E,n} = o_E, \bar{S}_n=s)1_{\bar{O}_{E,n}\in \mathbb{R}_{E}^{'}}\lvert \bar{S}_n=s\right]}\]
where $\bar{O}_{-E,n}$ is the vector of $E$ coordinates of $\bar{O}_n$ and similarly, $\bar{O}_{E,n}$ is defined.
Denoting 
 \[G_n^s(o_E)= \frac{1}{n}\log  \mathbb{P}(\bar{O}_{-E,n} \in \mathcal{R}_{-E}^{'}\lvert  \bar{O}_{E,n} = o_E, \bar{S}_n=s) = \frac{1}{n}\log \mathcal{B}(o_E, s)\]
 we know that $G_n^s(o_E)$ converges to a continuous function $G^s(o_E)$ uniformly on $o_E\in\mathcal{R}^{'}_E$ using a large deviation rate. Applying Lemma \ref{varadhan:lemma} gives
\begin{equation*}
\begin{aligned}
& \lim_n \frac{1}{n}\log\mathbb{P}\left(\bar{O}_n \in \mathcal{R}_O^{'} \lvert \bar{S}_n=s\right)\\
&={ \lim\limits_n \cfrac{1}{n}\log \mathbb{E}\left[ \exp\left(n \cdot \frac{1}{n}\log \mathcal{B}(\bar{O}_E,  s)\right)1_{\bar{O}_E\in \mathbb{R}_{E}^{'}}\lvert \bar{S}_n=s\right]}\\
&{= -\inf\limits\limits_{o_E \in \mathcal{R}^{'}_E}\left\{ \Lambda_{g_E}^*(D_E  s + P_E o_E +q_E/\sqrt{n})-G^s(o_E)\right\}.}
 \end{aligned}
\end{equation*}
Also, note that $G_n^s(o_E)$ is a sequence of continuous, concave functions in $o_E,s$ converging to a concave function $G^s(o_E)$. The concavity in $o_E$ follows from the fact that convolution of log-concave densities with a log-concave indicator preserves concavity. Thus, we have
\begin{equation*}
\begin{aligned}
&\inf\limits_{o_E\in \mathcal{R}^{'}_E}\left\{ \Lambda_{g_E}^*(D_E  s + P_E o_E +q_E/\sqrt{n})-G^s(o_E)\right\}\\
&=\lim\limits_n \inf\limits_{o_E\in \mathcal{R}^{'}_E}\left\{ \Lambda_{g_E}^*(D_E  s + P_E o_E +q_E/\sqrt{n})-G_n^s (o_E)\right\} \\
\end{aligned}
\end{equation*}
Denoting $H_n(s) = -\inf\limits_{o_E\in \mathcal{R}^{'}_E}\left\{ \Lambda_{g_E}^*(D_E  s + P_E o_E +q_E/\sqrt{n})-G_n^s (o_E)\right\}$, we finally note that $H_n(.)$ is a sequence of continuous functions with a uniform limit $H(.)$ on $\mathcal{R}_S$. This completes the proof of the second part with similar arguments as the first part of the theorem.
\end{enumerate}
\end{proof}

\noindent{\textbf{Proof of Lemma \ref{varadhan:lemma}}}
\begin{proof}
%Consider $N$ such that for all $n\geq N$
%\[ -\epsilon + H(z)  \leq H_n(z) \leq H(z) +\epsilon \text{ for all } z \in \mathcal{R} \text{ and}\]
%\[\frac{1}{n}\log\mathbb{E}[\exp(n H(\bar{Z}_n) 1_{\bar{Z}_n \in \mathcal{R}}]
Using the uniform convergence of $H_n(.)$ to $H(.)$ on compact $\mathcal{R}\in \real^d$, we have for any $x \in \mathcal{R}^{0}$ such that
$\Lambda^*(x) -H(x)<\infty$
and $\delta >0$ such that $\mathcal{B}(x, \delta) \subset \mathcal{R}^{0}$
\begin{equation*}
\begin{aligned}
&\liminf_n \frac{1}{n}\log\mathbb{E}[\exp(n H_n(\bar{Z}_n)) 1_{\bar{Z}_n \in \mathcal{R}}] \\
&= \liminf_n \frac{1}{n}\log\mathbb{E}\left[\exp(n (H_n(\bar{Z}_n)- H(\bar{Z}_n)) \exp(n H(\bar{Z}_n)) 1_{\bar{Z}_n\in \mathcal{R}}\right] \\
&\geq  -\lim_n \sup_{z\in \mathcal{R}} |H_n(z)- H(z)| + \liminf_n \frac{1}{n}\log\mathbb{E}\left[\exp(n H(\bar{Z}_n)) 1_{\bar{Z}_n \in \mathcal{R}^0}\right] \\
%&\geq \liminf_n \frac{1}{n}\log\mathbb{E}\left[\exp(n H_n(\bar{Z}_n)) 1_{\bar{Z}_n \in \mathcal{R}^{0}}\right] \\
&\geq \liminf_n \frac{1}{n}\log\mathbb{E}\left[\exp(n H(\bar{Z}_n)) 1_{\bar{Z}_n \in \mathcal{B}(x, \delta)}\right].
\end{aligned}
\end{equation*}
%Let 
%\[ \chi_{\mathcal{R}}(z) =\begin{cases} 
%      0 & \text{ if } z\in \mathcal{R} \\  
 %     \infty  & \text{ otherwise.}
%   \end{cases}
%\] 
Let $C>0$ be chosen such that it satisfies for $x\in \mathcal{R}^{0}$ 
\begin{equation}
\label{choice:C}
C > \Lambda^*(x) -H(x) + \liminf_n \frac{1}{n}\log\mathbb{E}\left[\exp(nH(\bar{Z}_n))1_{\bar{Z}_n \in \mathcal{B}^c(x, \delta)}\right].
\end{equation} 
where $\mathcal{B}^c(x, \delta)=\{z\in \real^d: z \notin \mathcal{B}(x, \delta)\}$.
Defining a continuous bounded function $\Psi(.)$ as 
\[\Psi(y) = C\cdot \min\left(\cfrac{d(x,y)}{\delta},1\right)\]
\begin{equation*}
\begin{aligned}
\mathbb{E}\left[ \exp(-n \Psi(\bar{Z}_n) + n H(\bar{Z}_n))\right] &\leq \exp(-n C) \cdot \mathbb{E}\left[ \exp(n H(\bar{Z}_n)) 1_{\bar{Z}_n \in \mathcal{B}^c(x, \delta)}\right]\\
&\;\;\;\;+ \mathbb{E}\left[ \exp(n H(\bar{Z}_n)) 1_{\bar{Z}_n \in \mathcal{B}(x, \delta)}\right].
\end{aligned}
\end{equation*}
Using Varadhan's limit lemma that states for $\bar{Z}_n$ satisfying a large deviation principle with a rate function $\Lambda^*(.)$
\[\lim_n \frac{1}{n}\mathbb{E}[\exp(-n \Psi(\bar{Z}_n) + n H(\bar{Z}_n))] = -\inf_{z} \{\Lambda^*(z) + \Psi(z)-H(z)\},\] 
and the above bound for $\mathbb{E}\left[ \exp(-n \Psi(\bar{Z}_n) + n H(\bar{Z}_n))\right]$, we have
\begin{equation*}
\begin{aligned}
&\max\Big(-C  + \frac{1}{n}\liminf_n \mathbb{E}\left[ \exp(n H(\bar{Z}_n)) 1_{\bar{Z}_n \in \mathcal{B}^c(x, \delta)}\right], \\
&\;\;\;\;\;\;\;\;\;\;\;\;\;\;\;\;\frac{1}{n}\liminf_n \mathbb{E}\left[ \exp(n H(\bar{Z}_n)) 1_{\bar{Z}_n \in \mathcal{B}(x, \delta)}\right]\Big)\\
& \geq \liminf_n \frac{1}{n}\log \mathbb{E}\left[ \exp(-n \Psi(\bar{Z}_n) + n H(\bar{Z}_n))\right]\\
&= -\inf_{z} \{\Lambda^*(z) + \Psi(z)-H(z)\}\geq -\Lambda^*(x) + H(x). 
\end{aligned}
\end{equation*}
Due to the choice of $C$ in \eqref{choice:C}, we can complete the proof of the lower bound on the limit of infimums by observing
\begin{equation*}
\begin{aligned}
\liminf_n \frac{1}{n}\log\mathbb{E}[\exp(n H_n(\bar{Z}_n)) 1_{\bar{Z}_n \in \mathcal{R}}] &\geq \liminf_n \frac{1}{n}\log\mathbb{E}\left[\exp(n H(\bar{Z}_n)) 1_{\bar{Z}_n \in \mathcal{B}(x, \delta)}\right] \\
&\geq -\Lambda^*(x) + H(x)\\
&\geq -\inf_{z\in \mathcal{R}^{0}}\{\Lambda^*(z)-H(z)\}\\
&=  -\inf_{z\in \mathcal{R}}\{\Lambda^*(z)-H(z)\}.
\end{aligned}
\end{equation*}

To prove the upper bound, let $\phi_j = j \min(d(z, \mathcal{R}),1)$ be a sequence of bounded continuous functions increasing to $\chi_{\mathcal{R}}(.)$, the characteristic function of $\mathcal{R}$. Again using the uniform convergence of $H_n(.)$ to $H(.)$ on $\mathcal{R}$ and Varadhan's limit lemma for continuous bounded function $\phi_j(.)$,
%\[\lim_n \frac{1}{n}\mathbb{E}[\exp(n \phi_j(\bar{Z}_n))] = -\inf_{z} \{\Lambda_f^*(z) - \phi_j(z)\},\]
we have
\begin{equation*}
\begin{aligned}
&\limsup_n \frac{1}{n}\log\mathbb{E}[\exp(n H_n(\bar{Z}_n)) 1_{\bar{Z}_n \in \mathcal{R}}) ] \\
&\leq \limsup_n \frac{1}{n}\log\mathbb{E}[\exp(n(H_n(\bar{Z}_n)- \phi_j(\bar{Z}_n))]\\
&\leq \limsup_n  \frac{1}{n}\log\mathbb{E}[\exp( n \sup_{z\in \mathcal{R}}|H_n(z)-H(z)| + n H(\bar{Z}_n) -n\phi_j(\bar{Z}_n)) ] \\
&= \lim_n  \sup_{z\in \mathcal{R}}|H_n(z)-H(z)|  + \lim_n \frac{1}{n}\log\mathbb{E}[\exp(n H(\bar{Z}_n)- n\phi_j(\bar{Z}_n))] \\
&= -\inf_{z} \left\{\Lambda^*(z) - H(z) + \phi_j(z)\right\}\\
&\leq -\liminf_{j\to \infty}\inf_{z} \left\{\Lambda^*(z) - H(z) + \phi_j(z)\right\}\leq -\inf_{z\in \mathcal{R}} \left\{\Lambda^*(z) - H(z)\right\}.
\end{aligned}
\end{equation*}
The penultimate step follows by applying Varadhan's lemma for a continuous, bounded function $H(.)- \phi_j(.)$.
The proof is thus complete by showing 
\[\liminf_{j\to \infty}\inf_{z} \left\{\Lambda^*(z) - H(z) + \phi_j(z)\right\} \geq \inf_{z\in \mathcal{R}} \left\{\Lambda^*(z) - H(z)\right\}.\]
Since $\phi_j(.) = 0$ on $\mathcal{R}$, it suffices to prove
\[\liminf_{j\to \infty}\inf_{z\in \mathcal{R}^c} \left\{\Lambda^*(z) - H(z) + \phi_j(z)\right\} \geq \inf_{z\in \mathcal{R}} \left\{\Lambda^*(z) - H(z)\right\}=\mathcal{L}.\]
Suppose the above claim is untrue. Then, there exists a subsequence $j_k$, $z_{j_k} \in  \mathcal{R}^c$ and $0<\epsilon<\mathcal{L}/2$ such that
\[\Lambda^*(z_{j_k}) - H(z_{j_k}) + \phi_{j_k}(z_{j_k})\leq \mathcal{L}-\epsilon.\]
This would in turn imply $d(z_{j_k}, \mathcal{R})\to 0$ as $k\to \infty$. This also means $\sup_{k}\Lambda^*(z_{j_k}) - H(z_{j_k}) \leq \mathcal{L}-\epsilon$. There exists a $z^*$ such that $\Lambda^*(z^*) - H(z^*) \leq \mathcal{L}-\epsilon$ and a further subsequence such that $d(z_{j_{k_l}}, z^*) \to 0$. For this subsequence, we know that we can construct a sequence $y_{j_{k_l}}\in \mathcal{R}$ such that 
$d(z_{j_{k_l}},y_{j_{k_l}}) \to 0.$ This would imply that $d(y_{j_{k_l}},z^*) \to 0$ and the consequence of this is that $z^* \in \mathcal{R}$ which shall lead to the contradiction \[ \Lambda^*(z^*) - H(z^*)> \mathcal{L} = \inf_{z\in \mathcal{R}} \left\{\Lambda^*(z) - H(z)\right\}.\]
\end{proof}

\section{Conjugates of barrier function}
\label{dual:details}
The barrier function for the canonical sign and cube constraints that we use in our implementations are 
\[b_{\mathcal{R}_{j,E}}(o_{j,E}) =\log\left(1 +\cfrac{1}{s_{j,E} o_{j,E}}\right) \text{ and }\] 
\[b_{\mathcal{R}_{j,-E}}(o_{j,-E}) = \log\left(1+ \cfrac{1}{\lambda - o_{j,-E}}\right) + \log\left(1+ \cfrac{1}{\lambda + o_{j,-E}}\right)\]
respectively. The conjugate for the sign barrier function at $P_{j,E}^T u$ can be computed as  
\[b_{\mathcal{R}_{j,E}}^*(v) = \sup_{z_j}z_j v - \log\left(1 +\cfrac{1}{s_{j,E} z_j}\right) \text{ at } v= P_{j,E}^T u.\]
The optimal $z_j^*$ maximizing the above optimization problem is given by
\[ z_j^*= \begin{cases} 
     -\cfrac{1}{2} + \sqrt{\cfrac{1}{4} -\cfrac{1}{v}} &\text{ if } v<0 \\
       \text{ no root } &\text{ otherwise}
   \end{cases}\]
whenever $s_{j,E}=1$ and
\[z_j^* = \begin{cases} 
     \cfrac{1}{2} - \sqrt{\cfrac{1}{4} + \cfrac{1}{v}} &\text{ if } v>0 \\
       \text{ no root } &\text{ otherwise}
   \end{cases}\]
whenever $s_{j,E}=-1$.
The conjugate for the cube barrier that reflect the inactive constraints does not have a closed form expression. Yet, we may employ an easy binary search method to find the roots of $p-|E|$ univariate separable conjugate problems
\[b_{\mathcal{R}_{j,-E}}^*(v) = \max_{z_j\in [-\lambda,\lambda]} v z_j -b_{\mathcal{R}_{j,-E}}(z_j) \text{ at } v= P_{j,-E}^T u.\]

We may alternatively choose a log-barrier on inactive coordinates in which case the barrier function is given by
\[b_{\mathcal{R}_{j,-E}}(o_{j,-E}) = -\log(\lambda - o_{j,-E}) - \log(\lambda + o_{j,-E})\] whose conjugate has an explicit form of roots within  $[-\lambda, \lambda]$. That is, solving the optimal point $z_j^*$ that yields
$$b_{\mathcal{R}_{j,-E}}^*(v)  =\max_{z_j \in [-\lambda, \lambda]} v z_j + \log(\lambda -z_j) + \log(\lambda + z_j) $$
is given by
\[z_j^*=\begin{cases} 
   -  \cfrac{1}{v} + \sqrt{\cfrac{1}{v^2} + \lambda^2} &\text{ if } v>0 \\[1em]
   - \cfrac{1}{v} - \sqrt{\cfrac{1}{v^2} + \lambda^2}      &\text{ if } v<0.
   \end{cases}\]

\section{Supplementary to data analysis: inference on causal variants}
\label{egene:data:appendix}
\textbf{Details of data}\\
We give below the details of the gene expression data set analyzed in \ref{egene:data} under Section \ref{experiments}. The data set consists of an outcome variable that represents gene expression levels of a gene with ID \textbf{``ENSG00000131697.13"}, sampled for $97$ individuals from the tissue of Liver. Both the gene expression outcome and the genotypic data consisting of local variants measured within 1 MB up and downstream from the transcription gene site are a part of the GTEx project \url{https://www.gtexportal.org/home/}. The DNA genotyping on blood-derived DNA samples of these individuals was performed at the GTEx Laboratory Data Analysis and Coordination Center (LDACC) at the Broad Institute. More details on the sample procurement, gene, variants inclusion and reads of gene-level expression can be found in the papers \cite{carithers2015novel, gtex2015genotype, aguet2016local}. The gene under study has been analyzed in \cite{aguet2016local} as part of an eQTL study. In fact, the mentioned paper aimed at discovering cis-eQTLs that are associations between local genetic variation and gene expression. This work also conducted a secondary analysis on genes that are believed to have at least one regulatory variant to identify potential causal variants. It is of natural interest for the biologist to be able to give reproducible estimates of the effect sizes of these discovered variants post a search over the set of all variants that leads to reporting/identifying the promising ones. 
We apply our methods in \ref{egene:data} to produces estimates for the effect sizes of these possible regulatory variants, chosen through a Lasso analysis.

\medskip
\noindent\textbf{Comparison of the Bayesian estimates with frequentist approach}
We have seen in Section \ref{experiments} that the Bayesian estimates under a flat prior and modeled along the conditional approach have good frequentist properties like coverage and risk. For interested readers, we also compute the adjusted frequentist estimates using the methods in \cite{panigrahi2017mcmc}. The mentioned paper uses a sampler free approach to solve for an intractable pivot and provides adjusted intervals and an approximate selective MLE based on the truncated likelihood, conditioned on the selection event. For the discovered SNPs post a randomized Lasso on prototype SNPs as described in \ref{egene:data}, we plot the frequentist estimates alongside the Bayesian estimates under the diffuse prior. The fact that the Bayesian intervals and posterior mean mimic the frequentist intervals and selective MLE validates that our Bayesian approach displays the Bernstein von Mises phenomenon as the unadjusted Bayesian estimates do.
\begin{figure}[h]
%\vskip 0.2in
\begin{center}
\centerline{\includegraphics[height=8cm,width=12cm]{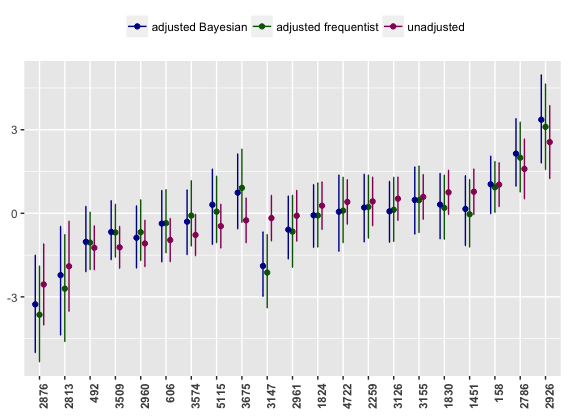}}
\caption{\small{Effect size estimates: adjusted posterior mean and credible intervals based on the truncated and unadjusted Bayesian posterior; selective MLE and confidence intervals are based on the sampler free approach in \cite{panigrahi2017mcmc}. The Bayesian estimates, both posterior mean and intervals mimic the frequentist estimates.}}
\label{egene:plot:freq}
\end{center}
%\vskip -0.2in
%\vspace{-1.5cm}
\end{figure} 

\end{document}